\documentclass[12pt]{article}

\usepackage[margin=1in]{geometry}
\usepackage{amsmath,amsthm,bm}
\usepackage{newtxtext,newtxmath}
\usepackage{microtype}

\usepackage{mathtools}

\usepackage{enumitem}
\usepackage{graphicx}
\usepackage{subcaption}
\usepackage{tikz}
\usetikzlibrary{patterns}
\usepackage[ruled]{algorithm2e}
\usepackage{natbib}
\bibpunct[, ]{(}{)}{,}{a}{}{,}

\usepackage{xcolor}
\usepackage{hyperref}
\usepackage[noabbrev,capitalize]{cleveref}
\crefname{equation}{}{}
\hypersetup{
  colorlinks=true,
  linkcolor=blue,
  citecolor=blue,
  urlcolor=blue
}

\theoremstyle{plain}
\newtheorem{theorem}{Theorem}
\newtheorem{lemma}{Lemma}
\newtheorem{proposition}{Proposition}
\newtheorem{corollary}{Corollary}

\theoremstyle{definition}
\newtheorem{definition}{Definition}
\newtheorem{property}{Property}

\newtheorem{myclaim}{Claim}

\providecommand{\Halmos}{}
\newcommand{\nonl}{\renewcommand{\nl}{\let\nl\oldnl}}

\newcommand{\paren}[1]{\left( #1 \right)}
\newcommand{\sqb}[1]{\left[ #1 \right]}
\newcommand{\set}[1]{\left\{ #1 \right\}}

\DeclareMathOperator*{\argmin}{arg\,min}

\definecolor{codered}{rgb}{.8,.1,.1}
\definecolor{codegreen}{rgb}{0,0.6,0}
\definecolor{codeblue}{rgb}{0,0.1,1}
\definecolor{codegray}{rgb}{0.5,0.5,0.5}
\definecolor{codepurple}{rgb}{0.58,0,0.82}

\newcommand{\Dcal}{\mathcal{D}}
\newcommand{\Ecal}{\mathcal{E}}

\newcommand{\Ical}{\mathcal{I}}

\newcommand{\Lcal}{\mathcal{L}}

\newcommand{\Ocal}{\mathcal{O}}
\newcommand{\Qcal}{\mathcal{Q}}
\newcommand{\Rcal}{\mathcal{R}}

\newcommand{\Ebb}{\mathbb{E}}
\newcommand{\Nbb}{\mathbb{N}}
\newcommand{\Pbb}{\mathbb{P}}

\newcommand{\Rbb}{\mathbb{R}}

\newcommand{\Rect}{\textnormal{LS}}
\newcommand{\Area}{\textnormal{Area}}
\newcommand{\DP}{\textnormal{DP}}

\newcommand{\negspace}{\hspace{-12em}}

\title{Approximation Algorithms for Inventory Problems with\\
Decomposable Submodular Ordering Costs}
\author{
Retsef Levi\thanks{Sloan School of Management, Massachusetts Institute of Technology. Email: \href{mailto:retsef@mit.edu}{retsef@mit.edu}.}
\and
Georgia Perakis\thanks{Sloan School of Management, Massachusetts Institute of Technology. Email: \href{mailto:georgiap@mit.edu}{georgiap@mit.edu}.}
\and
Emily Zhang\thanks{Operations Research Center, Massachusetts Institute of Technology. Email: \href{mailto:eyzhang@mit.edu}{eyzhang@mit.edu}.}
}
\date{}

\begin{document}
\maketitle

\begin{abstract}
This paper develops an approximation algorithm for the submodular joint replenishment problem (SJRP) under a broad family of decomposable submodular ordering cost functions. In the SJRP, a central planner coordinates orders to satisfy deterministic demand for multiple items over a finite discrete planning horizon while minimizing total holding and ordering costs, with the latter modeled as a submodular function of the subset of items ordered in each period. The ordering cost functions considered in this paper are defined based on a decomposition of the items into $k$ categories, where the cost is a function of weighted aggregate quantities within each category and allows for arbitrary interactions across categories through a joint cost function. The proposed algorithm rounds the solution to a linear programming relaxation by partitioning the fractional solution into nested regions according to marginal costs using a novel water-filling procedure, and then selecting one order from each region to obtain a feasible integral schedule. The resulting algorithm achieves an $O(k)$-approximation. When the number of categories $k$ is fixed, this yields the first constant-factor guarantee for this broad class of submodular ordering costs, significantly expanding the class of cost functions for which such guarantees are known.
\end{abstract}

\medskip
\section{Introduction}
The deterministic inventory problem studied in this paper captures several well-known models in the supply chain literature, including the economic lot size model \citep{wagner1958dynamic}, the inventory routing problem (IRP) \citep{bell1983improving}, and the submodular joint replenishment problem (SJRP) \citep{cheung2016submodular}. In all of these models, the decision maker must choose when and how much to order to satisfy known period-dependent demands over a finite discrete planning horizon with the goal of minimizing total ordering and holding costs. The IRP and SJRP are known to be NP-hard, but their approximability remains poorly understood due to the complexity of their combinatorial cost structures.

The economic lot size model considers a single item with deterministic, time-varying demand over $T$ periods. In each period, the decision maker selects an order quantity to satisfy current and future demand, where demand due in period~$d$ can be fulfilled by any order placed no later than its deadline~$d$. The goal is to determine a replenishment schedule that minimizes the total cost over all periods, which includes fixed ordering costs incurred whenever an order is placed, and holding costs proportional to the inventory carried between periods. The model captures a fundamental trade-off between frequent small orders (reducing holding costs) and infrequent large orders (reducing ordering costs), and has long served as a foundational framework in inventory theory. It admits efficient algorithms for computing an optimal cost-minimizing sequence of orders, beginning with the dynamic programming algorithm of \citet{wagner1958dynamic}.

The joint replenishment problem (JRP) generalizes the economic lot size model to a multi-item setting with $N$ items, where orders can be coordinated across items to reduce costs. The ordering cost function is defined through an item-specific cost and a fixed ordering (or setup) cost. Accordingly, the cost of each order is the sum of the fixed ordering cost (regardless of the items included) and the sum of the item-specific costs of all items included in the order, a cost structure which is called \textit{additive}. The JRP captures economies of scale from joint ordering but introduces combinatorial complexity in determining which subsets of items to order in each period under shared ordering costs.

The JRP is NP-hard \citep{arkin1989computational}, and a substantial body of work has focused on developing approximation algorithms for the problem \citep{levi2004primal, levi2008constant, bienkowski2014better}. The JRP has been used in many applications including inventory management, maintenance scheduling, logistics, and transportation problems \citep{cheung2016submodular}.

The submodular joint replenishment problem (SJRP), studied by \citet{cheung2016submodular}, generalizes the JRP by allowing the ordering cost in each period to be a monotone submodular function of the set of items ordered. This assumption includes the JRP as a special case, but offers significantly greater modeling flexibility. The submodular ordering cost can represent shared resources such as trucks, machines, or labor, and more richly captures economies of scale in joint ordering.

A closely related problem is the inventory routing problem (IRP), in which a central depot serves multiple geographically distributed customers. The delivery cost for each order corresponds to the length of the shortest tour that starts and ends at the depot while visiting the subset of customers served in the respective period. The objective is to minimize the total routing and holding costs over the planning horizon.

Both the SJRP and IRP are known to be computationally challenging and little is known about their approximability \citep{nagarajan2016approximation}. Beyond the JRP, constant-factor approximation algorithms have been developed only for special cases of the SJRP, including the cases when the ordering cost function is defined based on a tree, laminar, or cardinality structure \citep{bienkowski2014better, cheung2016submodular}. In the tree case, items are leaves in a rooted tree and the ordering cost of a subset is the sum of costs along paths from the root to the selected leaves. In the laminar case, the ordering cost is defined as the minimum total cost of a collection of nested groups that cover the item set. In the cardinality case, the cost is a concave function of the number of items ordered. However, the approximation algorithms rely on the cost function exhibiting very specific structure, and whether a constant-factor approximation exists for the SJRP with more general submodular ordering cost functions or for the IRP on general metrics remains a long-standing open question. A series of works has progressively improved worst-case approximation guarantees for the general SJRP and IRP \citep{nagarajan2016approximation, bosman2020improved}, with the current best \( O(\log \log \min(N, T)) \)-approximation ($N$ being the number of items and $T$ the planning horizon length) achieved by \cite{bosman2020improved}.

\subsection{Main contribution}

This paper develops an approximation algorithm for the SJRP under a broad family of submodular ordering cost functions that are based on a decomposition of the items into categories. These functions depend on a weighted aggregate quantity within each category and allow arbitrary interactions across categories through a joint cost function. 

Specifically, the class of submodular functions, formally defined in \cref{def:form_cost}, is constructed from a continuous DR-submodular function, item-specific weights, and a partition of the items into $k$ disjoint \textit{categories}. 
% For any given set of items, the function evaluates the weighted sum of items within each part of the partition and applies the concave function to the resulting $k$-dimensional vector.
The partition and the item-specific weights are used to associate a vector with each subset of items $S$, consisting of the sum of the weights of the items included in the intersection of $S$ with each of the item categories. The cost of ordering $S$ is obtained by applying a continuous DR-submodular function to this vector.

% Intuitively, each part of the partition represents a distinct \textit{category} of items. 
This class of functions enables flexible and realistic modeling of cross-category interactions, as the continuous function is not required to be separable, allowing it to capture complex dependencies between categories. Within each category, the use of weighted sums reflects the natural assumption that the costs depend on weighted aggregate quantities---such as total weight, volume, or value---with the item-specific weights reflecting heterogeneous resource consumption across items. This formulation is particularly well-suited to supply chain and logistics applications, where cross-category cost interactions commonly arise due to shared storage, transportation, or operational constraints.

It is worth noting that the newly proposed class of submodular functions captures many well-studied special cases of submodular cost functions. For example, when the decomposition consists of a single part ($k = 1$), it recovers both the JRP (where the additive ordering cost is the sum of item-specific weights plus a fixed base cost) and cardinality-based costs (where the ordering cost depends only on the number of items in the order, regardless of their identities). %While constant-factor approximation algorithms were previously known for each of these cases separately, our framework captures both simultaneously when $k = 1$ and provides a unified approach to obtain a constant-factor guarantee. 

The main result of the paper shows that the SJRP with the class of submodular ordering cost functions described above admits an efficient approximation algorithm, with a worst-case approximation guarantee that scales with the number of item categories $k$. Specifically, we construct an algorithm 
% (Algorithm~\ref{alg:main}, described in \cref{section: algorithm}) 
that runs in polynomial time and achieves an $(8k+1)$-approximation guarantee. This yields a \textit{constant-factor approximation} when $k$ is a constant---a guarantee that was previously unknown even in the special case of $k = 1$. This result significantly expands the class of cost functions for which constant-factor guarantees are known. Even for small $k$, the cost functions considered are highly expressive, as the decomposition permits arbitrary cross-category interactions. Additionally, the structural insights obtained as part of the analysis in this paper may inform future progress for the SJRP and the IRP.

%It includes general submodular functions when each \( P_j \) is a singleton (i.e., \( k = N \)). 

%\cref{lemma: f conditions} provides conditions on the function $f$ under which the cost function $K$, written in the format from \cref{def:form_cost}, is submodular. 

% Even for a constant value of $k$, the cost functions we consider are highly expressive, as the decomposition function $f$ can encode arbitrary interactions among its arguments. To our knowledge, this is the first constant-factor approximation result that accommodates such general interactions. 

% Does not cover the following submodular functions (because the cost of an item depends on the membership of other items):
% \begin{itemize}
% \item Laminar: (don't think so)
% \item Tree: Let the subsets be all of the sets of the children of all nodes in the tree. Need $O(N)$ subsets because all of the leaves are not exchangeable. 
% \end{itemize}

\subsubsection{Algorithm Overview}

The approximate solution to the SJRP is obtained by first constructing a corresponding instance of the \emph{Disjoint Interval Covering Problem} (DICP). In the DICP, there is a collection of disjoint \textit{demand intervals} over the horizon for each item, and the objective is to place orders so that each demand interval includes at least one order with the respective item. This formulation removes the need to consider holding costs, focusing the objective solely on minimizing ordering costs. Unlike prior approaches that eliminate holding costs without additional structural constraints, the DICP enforces disjoint demand intervals for each item (that is, they do not overlap along the time horizon), simplifying the structure of the problem and facilitating algorithmic design. The approximation algorithm for the SJRP is obtained by leveraging a novel approximation algorithm for the DICP as a subroutine. It is shown that by doing so, there is only a constant factor loss in the approximation to the SJRP.

The DICP approximation algorithm is the primary technical contribution of this work. The algorithm begins by solving the linear programming relaxation of the DICP and then applies a rounding procedure. The approach exploits structural properties of the LP solution and the specific class of submodular ordering cost functions considered. The algorithm's design and analysis further rely on a notion of \textit{level} for subsets of items relative to a given category, which intuitively captures the marginal cost of adding one more item from that category to the subset. These levels induce a total order over the subsets ordered by the optimal LP solution. A water-filling procedure then partitions the LP solution into regions that group subsets with similar levels. The first region contains the lowest-level subsets, and subsequent regions are formed by progressively incorporating subsets with higher levels.

% , which is used to divide them into regions. Each region groups together subsets with similar levels and is constructed through a water-filling procedure: the first region includes the subsets with the lowest levels, and subsequent regions are formed by progressively adding higher-level subsets.

An integral schedule is obtained by solving a computationally tractable dynamic program to select one subset to order from each constructed region, minimizing the total cost. Leveraging the LP’s optimality conditions, it is shown that ordering at least one subset from each region suffices to satisfy all demands.

A key structural insight is that, for each demand point, there exists a level~$\ell$ such that every subset ordered by the LP in its demand interval has level less than or equal to~$\ell$ if and only if it contains the corresponding item. This observation, formalized in \cref{lemma: consistent levels}, is central to the correctness of the rounding scheme. These structural properties may also be useful in the design of future algorithms for the SJRP and related problems with submodular costs.

% One of the main points of constructing the regions via the water-filling procedure is that these will form a particular nested structure, well-suited for satisfying these regions via a randomized procedure. More precisely, we will satisfy regions starting from lower-level regions as follows: for any region that is not yet satisfied, we randomly select a subset from this region with probabilities given by the LP. As a note, proceeding by increasing order of the level $L_j$ ensures that larger subsets are ordered first since these also fulfill corresponding regions from higher levels.

\subsection{Literature Review}

% Most of the literature on JRP focuses on the additive ordering cost structure, where placing an order incurs a fixed joint cost plus item-specific fixed costs. This version of the problem 
The JRP was shown to be NP-hard by \cite{arkin1989computational} and APX-hard by \cite{nonner2009approximating}, which means that it does not admit a polynomial-time approximation scheme (PTAS)---that is, a polynomial-time algorithm with approximation ratio $(1 + \epsilon)$. A sequence of increasingly tighter approximation algorithms has been developed for this setting. The current best approximation ratios are 2 due to \cite{levi2004primal}, 1.80 due to \cite{levi2008constant}, and 1.791 due to \cite{bienkowski2014better}. In contrast, this paper studies the submodular JRP, which features a more general and expressive cost structure, and for which constant-factor approximations are not yet known.

A closely related problem is the IRP, introduced by \cite{bell1983improving}, which modifies the JRP by defining setup costs as the minimum routing cost to serve a set of customers in a metric space. The IRP has been extensively studied in operations research and logistics \citep{burns1985distribution, federgruen1984combined}; see \cite{coelho2014thirty} for a comprehensive review of the literature. \cite{fukunaga2014deliver} provided a constant-factor approximation under the restriction that orders are scheduled periodically, but the approximability of general ordering policies is not known. The IRP shares the same integer programming formulation as the JRP \citep{bosman2020improved}, differing only in the form of the joint setup cost.% As a result, many works develop approximation algorithms for both problems in parallel \citep{bosman2020improved, nagarajan2016approximation}.

The SJRP was studied by \cite{cheung2016submodular}. In contrast with the JRP, it remains open whether the SJRP admits a constant-factor approximation in general. \cite{cheung2016submodular} provide constant-factor approximation algorithms for special cases of submodular cost functions, including tree, laminar, and cardinality costs. For general submodular cost functions, \cite{nagarajan2016approximation} established the first sub-logarithmic approximation guarantee, assuming that all holding costs are fixed-degree polynomial functions of the holding time. Their algorithm achieves an \( O(\log T / \log \log T) \)-approximation for both the SJRP and the IRP. This was later improved exponentially by \cite{bosman2020improved}, who obtained an \( O(\log \log \min(N, T)) \)-approximation. In contrast, this paper develops an approximation algorithm for a broad family of submodular functions, with an approximation guarantee that depends only on the number of item categories.

The submodular function decomposition studied in this paper generalizes the decomposable submodular functions introduced by \citet{stobbe2010efficient}. In their formulation, the cost of a given set of items is expressed as a sum of concave functions, each applied to a weighted sum of item indicators, with one concave function and weight vector per category of items. When the weight vectors have disjoint supports—that is, each item belongs to only one category—their model becomes a special case of the decomposition considered in this paper. The class of functions considered in this paper is significantly more general, as it allows for arbitrary interactions across categories through a joint cost function, whereas their model assumes an additive structure with separable components. 

\cite{stobbe2010efficient} demonstrate that such functions arise naturally in a range of applications, including classification tasks and multiclass queuing systems. They further develop efficient algorithms that exploit the decomposable structure, significantly improving scalability compared to generic submodular minimization methods. The decomposable framework thus strikes a useful balance between modeling power and tractability.

\subsection{Paper Organization}
The remainder of the paper is organized as follows. \cref{section: model} presents a mathematical formulation for the SJRP. \cref{section: algorithm} presents the approximation algorithm for the simpler DICP, and \cref{section: correctness} proves its correctness. \cref{section: Reduction} then describes how this algorithm is used as a subroutine to approximate the original SJRP. Finally, we conclude in \cref{section: conclusion}.

\section{Model} \label{section: model}

This section formally defines the submodular joint replenishment problem (SJRP) and presents a mathematical formulation, following the frameworks of \cite{bosman2020improved,cheung2016submodular}, and \cite{nagarajan2016approximation}. %The two models share the same underlying structure and differ only in the form of the joint setup cost---in the IRP it captures the minimum routing cost to serve a set of demand points in a metric space, whereas in the SJRP it captures a submodular ordering cost.
Throughout the paper, $[k]$ denotes the set $\{1, \dots, k\}$ for any integer $k$, and $[k, k']$ for any integers $k \leq k'$ denotes $\{k, k+1, \dots, k'\}$. 
For any subset $S \subseteq [N]$, let $\mathbf{e}_S \in \mathbb{R}^N$ denote its indicator vector, where the $i$th coordinate is $1$ if $i \in S$ and $0$ otherwise.

Consider a finite planning horizon of $T$ discrete time periods, indexed by $t \in [T]$, and a set of $N$ items, indexed by $i\in [N]$. For each item $i\in[N]$ and time period $d\in[T]$, $q_{id}\geq 0$ denotes the demand quantity of item $i$ that must be ordered by the deadline $d$. Let $\mathcal D$ denote the set of demand points $(i,d)$ with $q_{id} >0$. Each such demand point $(i,d)\in\mathcal D$ must be fulfilled by an order including item $i$ placed no later than its deadline $d$.

Let $h^i_{sd} \geq 0$ denote the cost of holding one unit of item $i$ from period $s$ to $d$. Serving demand point $(i,d)$ using an order placed at time period $s$ incurs a holding cost of $q_{id}h^i_{sd}$. Note that per unit production costs for ordering a unit of an item $i$ in a period $s$ can be captured by this framework through the holding cost $h^i_{sd}$. The assumption is that $h^i_{sd}$ is nonnegative and non-increasing in $s$ for fixed $i$ and $d$, following the typical setting where holding items in inventory for longer incurs higher costs. This implies that in an optimal policy, each demand point $(i,d)$ is served by the latest order including item $i$ placed at or before the deadline period $d$.

The decision maker has to determine which items to order in each time period, with the goal of minimizing the total cost, which consists of the joint ordering and the holding costs for carrying inventory until the demand deadline. Notably, the demand magnitude influences only the holding cost, as the ordering cost depends solely on the set of distinct items ordered, not on the quantity ordered.

Let the function $K(S)$ denote the cost of ordering a subset of items $S\subseteq[N]$, and assume that $K: 2^{[N]} \rightarrow \mathbb R_+$ is nonnegative and monotone increasing. 
% In the IRP, demand points are embedded in a metric space along with a designated depot, and the joint setup cost $K(S)$ is defined as the length of the shortest tour that visits the depot and each item in $S$. In contrast, under the SJRP, 
$K(S)$ is further assumed to be a submodular set function. That is, for any item $i\in [N]$, the marginal cost of adding item $i$ to a smaller subset is at least as large as the cost of adding it to a larger one, i.e., for all $S_1\subseteq S_2 \subseteq [N]$ and $i\notin S_2$, we have $K(S_1 \cup \set{i}) - K(S_1)\geq K(S_2 \cup \set{i}) - K(S_2)$. This paper focuses on cost functions $K$ that admit a structured decomposition, formally defined as follows:

\begin{definition}[Submodular function decomposition]
\label{def:form_cost}
Let $P_1, P_2,\dots, P_k$ be a partition of the items $[N]$ into $k$ disjoint subsets. 
Let $\mathbf{w} \in \mathbb{R}_{>0}^N$ denote a vector of item-specific weights, and define $\mathbf{w}_j := \mathbf{w}\circ \mathbf e_{P_j}$ for $j\in[k]$, where $\circ$ denotes the Hadamard (elementwise) product. We consider ordering cost functions $K: 2^{[N]} \rightarrow \mathbb{R}_+$ of the form:
\begin{align} \label{eq:form_cost}
    K(S) = f\paren{\mathbf w_1^\top \mathbf e_{S},\mathbf w_2^\top \mathbf e_{S},\dots, \mathbf w_k^\top \mathbf e_{S}},
\end{align}
where $f:\mathbb{R}_+^k \to \mathbb{R}_+$ is monotone nondecreasing, twice continuously differentiable, and DR-submodular, i.e.,
\[
\frac{\partial^2 f}{\partial x_i \partial x_j}(\mathbf{x}) \le 0 
\quad \text{for all } i,j\in[k] \text{ and } \mathbf{x} \in \mathbb{R}_+^k.
\]
\end{definition}

The assumption on $f$ in \cref{def:form_cost} ensures that the resulting ordering cost function $K$ is submodular (see \cref{lemma: f conditions}). For notational convenience, given a submodular cost function $K$ of the form defined in \eqref{eq:form_cost}, define the vector $\mathbf u(S) := \paren{\mathbf w_j^\top \mathbf e_{S}}_{j\in[k]} \in \mathbb R^k$ for any subset of items $S\subseteq [N]$. That is, $\mathbf u(S)$ is the vector of weighted sums over each category, and the cost function satisfies $K(S) = f\paren{\mathbf u(S)}$.

The submodular function decomposition introduced in \cref{def:form_cost} captures a broad class of well-studied submodular cost functions. To illustrate, consider an additive cost structure in which placing an order incurs a fixed cost $K_0$, and ordering item $i \in [N]$ incurs an additional additive cost $K_i$. Let $K_{\min} := \min_{i \in [N]} K_i$. This cost can be represented within our framework by setting $k = 1$, letting $\mathbf{w}_1 = (K_i)_{i \in [N]}$, and choosing a concave function $f : \mathbb{R} \to \mathbb{R}$ that passes through the points $(0,0)$ and $(K_{\min}, K_0 + K_{\min})$, and is linear with slope 1 beyond $K_{\min}$. As another example, cardinality-based cost functions are recovered by setting all weights equal to one and choosing $f$ to be any concave function.

% \begin{itemize}
%     % \item $K(S) > 0$ for all non-empty subsets $S\subseteq[N]$ and $K(\emptyset) = 0$,
%     % \item for every $S_1\subseteq S_2 \subseteq [N]$, we have $K(S_1)\leq K(S_2)$, and
%     \item for every $S_1\subseteq S_2 \subseteq [N]$ and any item $i\in [N]$, we have $K(S_1 \cup \set{i}) - K(S_1)\geq K(S_2 \cup \set{i}) - K(S_2)$. 
% \end{itemize}

A natural integer programming formulation for the SJRP is described in \cref{eq: IP} below. The binary variable $y_s^S$ is equal to 1 if the subset $S\subseteq [N]$ is ordered in period $s$, and 0 otherwise. The binary variable $x_{sd}^i$ is equal to 1 if demand $(i,d)$ is fulfilled using an order placed in period $s$, and 0 otherwise.

\begin{align}\label{eq: IP}
  \text{minimize}\quad  &\sum_{S\subseteq [N] }\sum_{s=1}^T K(S) y^S_s + \sum_{(i,d)\in \mathcal D} \sum_{s=1}^d q_{id}h^i_{sd} x^i_{sd} \tag{IP}&&\\
    \text{subject to}\quad&
     \sum_{s=1}^d x^i_{sd} = 1, &&\negspace
     (i,d)\in \mathcal D \tag{Constraint 1}\\
    & x^i_{sd} \leq \sum_{S: i\in S \subseteq [N]}y_s^S, &&\negspace(i,d)\in \mathcal D, s\in [d]  \tag{Constraint 2}\\
    &x^i_{sd}, y_s^S \in\set{0,1}, &&\negspace(i,d)\in \mathcal D, s\in [d], S\subseteq [N]\nonumber
\end{align}

In the formulation described in \eqref{eq: IP}, Constraint 1 ensures that each demand $(i,d)\in \Dcal$ is served by an order placed at some period $s\leq d$. Constraint 2 ensures that if a demand $(i,d)$ is served by an order placed at some period $s$, then there must be a subset ordered at period $s$ that includes item $i$.

Let (LP) denote the linear programming relaxation of \cref{eq: IP}, obtained by replacing the integrality constraints with nonnegativity constraints. Although (LP) contains an exponential number of variables, its dual can be solved in polynomial time using the ellipsoid method, since a separation oracle for identifying violated dual constraints can be implemented using submodular function minimization \citep{cheung2016submodular}. A corresponding primal solution can then be recovered from the dual solution in polynomial time \citep{schrijver1998theory}. %We obtain an approximate solution to \eqref{eq: IP} by rounding the optimal solution to (LP).

\subsection{Special Cases of the SJRP}

This section discusses several special cases of the SJRP. Some of these special cases have been studied in prior literature that is focused on models with simpler cost functions. These variants are described below, and it is subsequently explained in \cref{section: Reduction} how approximation algorithms developed for the simpler models can be leveraged to approximate solutions for the general SJRP.

\paragraph{SJRP with Deadlines.}
The \textit{SJRP with deadlines} is a special case of the SJRP in which holding costs are eliminated, and each demand must be fulfilled within a specified time interval, referred to as the demand interval. Following \citet{bienkowski2015approximation} and \citet{nonner2009approximating}, each demand for an item \( i \in [N] \) is represented by a triple \( (i, c, d) \) with \( c \leq d \in [T] \), where $c$ is when the demand interval commences and $d$ is its deadline. That is, item \( i \) must be ordered during the interval \( [c, d] \). The input consists of a set \( \mathcal{D} \) of such demands, and the goal is to place orders that satisfy each within its designated interval, while minimizing ordering costs. This formulation can be viewed as a special case of the SJRP in which the holding cost \( h^i_{sd} \) is zero when $s$ is within the associated demand interval and infinite otherwise. Prior work has shown that the SJRP can be reduced to the deadline variant with only a constant-factor loss in the approximation guarantee \citep{cheung2016submodular, nagarajan2016approximation}.

\paragraph{Disjoint Interval Covering Problem (DICP).}
The disjoint interval covering problem (DICP) is a variant of the SJRP with deadlines in which the demand intervals for each item are mutually \emph{disjoint}—that is, they do not overlap along the time horizon. Formally, each item $i\in[N]$ is associated with a collection of pairwise disjoint time intervals $\set{[c,d]: (i,c,d)\in \mathcal D}$. The objective remains the same as in the SJRP with deadlines: to minimize the total joint ordering cost while ensuring that each item is ordered at least once within each of its associated demand intervals.

The DICP can be formulated as an integer program, where binary decision variables $y_s^S$ indicate whether a subset of items $S$ is ordered at time $s$. Below is a linear programming relaxation of the DICP, in which the binary constraints are replaced with nonnegativity constraints. 
\begin{align}\label{eq: primal}
\text{minimize}\quad&\sum_{s=1}^T\sum_{S\subseteq [N] } K(S) y^S_s \tag{DICP}\\
    \text{subject to}\quad
    & \sum_{s\in [c,d]} \sum_{S: i\in S }y_s^S \geq 1, && (i,c,d)\in \mathcal D \nonumber\\
    &y_s^S \geq 0 &&s\in [T], S\subseteq [N].\nonumber
\end{align}

\section{DICP Approximation algorithm}
 \label{section: algorithm}

This section presents an approximation algorithm for the simpler DICP. \cref{section: Reduction} then describes how this algorithm is used as a subroutine to approximate the original SJRP.

The algorithm takes as input the LP solution $(y_t^S)_{t\in[T],S\subseteq [N]}$ to the DICP relaxation \eqref{eq: primal} and uses it as the basis for a rounding algorithm that yields a $4k$-approximation to the optimal integral DICP solution, where $k$ is the number of item categories or parts in the partition as in \cref{def:form_cost}. The output of the rounding algorithm is a feasible integral solution---specifically, a sequence of orders $O_t\subseteq [N]$ placed at each time period $t\in[T]$ such that each item $i\in[N]$ is ordered in each of its associated demand intervals. At a high level, the algorithm uses a water-filling partitioning procedure to group fractionally ordered sets in the LP solution into regions based on  appropriately defined marginal costs, and then applies a dynamic program to select orders within these regions to satisfy all demand while minimizing cost. The algorithm is shown to be feasible (\cref{subsection: feasibility of alg for decomposition}) and achieves a total ordering cost at most $4k$ times the LP optimum (\cref{subsection: bounding cost}).

% \subsection{Preliminaries}
% \label{subsection: prelims}

% This section establishes two key structural properties of the optimal solution \( (y_t^S)_{t \in [T], S \subseteq [N]} \) to the LP relaxation of \eqref{eq: primal}, which are essential for the design and analysis of the approximation algorithm. The fractional LP solution is interpreted as a probabilistic schedule of orders: at each time $t$, a subset $S$ is ordered with probability $y_t^S$. 

The approximation algorithm will leverage two structural properties of the optimal solution of \eqref{eq: primal}. Property 1 below follows from \citet{cheung2016submodular} and shows that the fractional orders in each period can be assumed, without loss of generality, to form a collection of nested sets. %This structural property is formalized below.

\begin{property}[Nested structure property]
\label{property: nested}
    For any time period $t\in[T]$ and subsets $R, S\subseteq [N]$, if $y^R_t>0$ and $y^S_t>0$, then either $R\subseteq S$ or $S\subseteq R$. 
\end{property}

Accordingly, for a given time period $t\in[T]$, the support of the fractional orders placed by the optimal solution of \eqref{eq: primal}, $\set{S\subseteq [N]: y_t^S>0}$, can be indexed as 
\begin{align}\label{eq: rewrite_LP_solution}
    S_1(t)\supsetneq S_2(t) \supsetneq \ldots \supsetneq S_{n(t)}(t) \supsetneq S_{n(t)+1}(t):=\emptyset,
\end{align}
where $n(t)$ is the number of different sets with positive fractional order quantity at time $t$. To simplify notation, for each $r\in [n(t)]$, let $y_r(t):=y_t^{S_r(t)}$ be the fractional order quantity of the subset $S_r(t)$ at time $t$.

% To obtain theoretical guarantees for the DICP approximation algorithm, we further require a second structural property of the optimal LP solution, stated in \cref{property:LP_solution_reduced}. 
% This property can be achieved through a pre-processing step that modifies the LP solution \( (y_t^S)_{t \in [T], S \subseteq [N]} \) without affecting its feasibility or optimality for \eqref{eq: primal}. %The transformed solution satisfies both \cref{property: nested} and \cref{property:LP_solution_reduced}. %Since \cref{property:LP_solution_reduced} is not needed for understanding the main algorithm, we defer its introduction to \cref{subsection: preprocessing} for readability.

% \subsection{Preprocessing} \label{subsection: preprocessing}

% As an initial pre-processing step, we transform the LP solution \( (y_t^S)_{t \in [T], S \subseteq [N]} \) so that, in addition to being optimal for \eqref{eq: primal} and satisfying \cref{property: nested}, it also satisfies an additional structural property which allows for consistent tie-breaking across time steps. We first introduce some notation. 

% Now, for each item $i \in S_1(t)$, let $r(i) \in [n(t)]$ denote the largest index such that $S_{r(i)}(t)$ contains item $i$. Formally,
% \[
% r(i) := \max \left\{ r\in [n(t)] : i \in S_r(t) \right\}.
% \]
% For convenience, define $B_t^i := S_{r(i)}(t)$, and define $A_t^i := S_{r(i)+1}(t)$. 

For any period $t\in[T]$ and any item $i\in[N]$, let $A^i_t$ denote the largest subset in $\set{S_1(t), \dots, S_{n(t)+1}(t)}$ that does not contain $i$. If $i\in S_1(t)$, let $B^i_t$ denote the smallest subset in this collection that contains $i$. The sets $A^i_t$ and $B^i_t$ are adjacent in the ordering. That is, if $B^i_t=S_r(t)$, then $A^i_t=S_{r+1}(t)$. 
% , or the full item set $[N]$ if no such subset exists. 
% In other words, $B_t^i$ is the smallest subset in the sequence $\{S_r(t)\}_{r \in [n(t)]}$ that contains $i$, and $A_t^i$ is the largest subset in $\{S_r(t)\}_{r \in [n(t)+1]}$ that does not contain $i$. 

% Notably, the LP optimality conditions imply the following constraint on the marginal costs of adding items to the subsets ordered in the solution.
Lemma 1 below states a property of the optimal fractional solution of the LP:
\begin{lemma}\label{claim:optimality_condition_dual_variable}
The optimal solution satisfies the following property. For any demand point $(i,c,d)\in\mathcal D$, the following holds:
    \begin{align*}
\max_{t\in [c,d]: i\in S_1(t)}& \sqb{ K\paren{B^i_t} - K\paren{B^i_t\setminus \set{i}}}\\
\leq 
&\min_{t\in [c,d]} \sqb{ K\paren{A^i_t \cup \set{i}} - K\paren{A^i_t}}.
    \end{align*}
\end{lemma}

\cref{claim:optimality_condition_dual_variable} captures the intuition that the optimality of the fractional solution implies that the marginal cost of increasing the amount of item $i$ ordered at any time within its demand interval must be no less than the marginal cost of reducing the amount ordered at any other time in that interval with positive fractional order of $i$. Otherwise, one could shift mass across time periods to reduce cost, contradicting the optimality of the LP solution. It turns out that the inequalities in \cref{claim:optimality_condition_dual_variable} can be made strict by adjusting the fractional order quantities while preserving optimality, which will be important for establishing the theoretical guarantees. This is formalized in the following property: 
% The desired property is stated formally below: 

\begin{property}[Temporal tie-breaking property]
\label{property:LP_solution_reduced}
    There exists an LP optimal solution such that for any demand point $(i,c,d)\in\Dcal$ and any two distinct times $t_1< t_2$ within its demand interval, either item $i$ is not ordered by the LP solution in period $t_2$ (i.e., $\sum_{S:i\in S} y_{t_2}^S=0$), or $$  K(B_{t_2}^i) - K(B_{t_2}^i\setminus\{i\})< K(A_{t_1}^i \cup\{i\}) - K(A_{t_1}^i).$$
\end{property}

Intuitively, Property~\ref{property:LP_solution_reduced} prioritizes ordering subsets in earlier times $t_1<t_2$. Suppose instead that $ K(B_{t_2}^i) - K(B_{t_2}^i\setminus\{i\})\geq K(A_{t_1}^i \cup\{i\}) - K(A_{t_1}^i) $. Then by \cref{claim:optimality_condition_dual_variable}, the equality $K(A_{t_1}^i \cup\{i\}) - K(A_{t_1}^i) = K(B_{t_2}^i) - K(B_{t_2}^i\setminus\{i\})$ must hold. 
In this case, it is possible to increase the amount of item \( i \) ordered at time \( t_1 \) and decrease the amount ordered at \( t_2 \) without increasing the LP cost. The role of Property~\ref{property:LP_solution_reduced} is to ensure that, in the presence of such ties, item \( i \) is ordered at the earlier time \( t_1 \). This provides a consistent tie-breaking convention across time steps, which is an important invariant for the algorithm, as it relies on a unified ordering of subsets ordered in the LP solution. 
It is shown in \cref{lemma: exists_satisfying_properties} of \cref{appendix:preprocessing} that there exists an optimal solution to \eqref{eq: primal} satisfying both \cref{property: nested} and \cref{property:LP_solution_reduced}.

% Due to the submodular cost structure, the fractional subsets ordered at each time period can be adjusted (without increasing ordering cost) to have a nested structure: for any time period $s$ and subsets $R, S\subseteq [N]$, if $y^R_s>0$ and $y^S_s>0$, then either $R\subseteq S$ or $S\subseteq R$. The proof of this statement and a method for obtaining such a solution by adjusting an optimal LP solution are established in Lemma 1 of \cite{cheung2016submodular}.

% \subsection{Algorithm}\label{subsection: DICP algorithm}
\paragraph{Algorithm overview. } The algorithm operates separately on each item category $P_j$ for $j\in[k]$ and consists of two main steps. The first step applies a water-filling procedure to all subsets ordered across all time periods by the optimal LP solution to \eqref{eq: primal}, $\set{S_r(t),t\in[T],r\in[n(t)]}$. Subsets are processed in order of their \textit{levels} with respect to a given category $j$, introduced in \cref{subsection: level}. The water-filling procedure, described in \cref{subsection: water filling}, constructs a collection of regions from the solution to \cref{eq: primal}. The second step, detailed in \cref{subsection: placing orders}, places orders within each constructed region using a dynamic program. The procedures described in \cref{subsection: water filling} and \cref{subsection: placing orders} are performed separately for each category $j\in[k]$. The final order schedule is then constructed by consolidating the resulting orders across all categories, as detailed in \cref{subsection: final schedule}.

\subsection{Subset levels.}\label{subsection: level}
The water-filling algorithm described subsequently, relies on the notion of subset \emph{levels} that is defined for each subset $S$ of items and item category $j\in[k]$, which determine the water level in the water-filling procedure. Intuitively, the level represents the marginal cost of adding an additional item from a given category to the subset. %These levels induce a total order over the subsets ordered by the optimal LP solution, which is used to divide them into regions.
Let $\frac{\partial f}{\partial u_j}$ denote the partial derivative of $f$ with respect to its $j$-th argument.

\begin{definition}[The $j$-level of subsets]
\label{def:level}
    For any category $j\in[k]$ and subset of items $S\subseteq [N]$, let $i^\star \in \argmin\set{w_i, i\in P_j\setminus S}$ and let \( \mathbf{e}_j \in \mathbb{R}^k \) denote the unit vector in direction \( j \). The \textit{$j$-level} of $S$ is defined as 
    \begin{equation*}
        L_j(S) := \begin{cases}
            \frac{K(S \cup \{i^\star\})-K(S)}{w_{i^\star}} & \text{if } P_j\not\subseteq S,\\
            0
             & \text{otherwise.}
        \end{cases}
    \end{equation*}
\end{definition}

To simplify the exposition, when the category $j$ is clear from context, the $j$-level of $S$ will be simply called the \emph{level} of $S$. Note that for any $S\subseteq [N]$ with $P_j \not\subseteq S$:
\begin{align*}
 L_j(S) &= \frac{K(S \cup \{i^\star\})-K(S)}{w_{i^\star}} \\
 &= \frac{f(\mathbf u(S) + w_{i^\star} \cdot \mathbf e_j) - f(\mathbf u(S))}{w_{i^\star}} \\
 &\approx  \left. \frac{\partial f}{\partial u_j} \right|_{u = \mathbf u(S)}.
\end{align*}
Therefore, $L_j(S)$ approximately captures the marginal cost of ordering an additional item with weight 1 from category $P_j$. That is, ordering an item from $P_j$ with weight $w$ would incur marginal cost approximately $w \cdot L_j(S)$. While the precise definition of the $j$-level is somewhat involved, its precise form is not critical for describing the algorithm. The key idea is that these values serve as \emph{levels} for comparing different subsets ordered by the LP solution across time periods.

Note that, by the concavity of $f$ in its $j$-th coordinate, the quantity $$\frac{f(\mathbf u(S) + w \cdot \mathbf e_j) - f(\mathbf u(S))}{w}$$ is non-increasing in $w$. Hence, similarly to marginal costs, the $j$-levels satisfy a natural monotonicity property:

\begin{lemma}\label{lemma:monotonicity_levels}
    For any subset of items $S_1\supseteq S_2$, the $j$-level satisfies $L_j(S_1) \leq L_j(S_2)$ for any $j\in [k]$.
\end{lemma}
% This monotonicity property reflects the intuition that marginal costs decrease as more items are added to a set. 
In particular, combined with Property~\ref{property: nested}  of the optimal LP solution, it follows that for each period $t\in[T]$:
\begin{equation*}
    L_j(S_1(t)) \leq L_j(S_2(t)) \leq\ldots \leq L_j(S_{n(t)}(t)). %,\quad t\in[T].
\end{equation*}

The algorithm depends only on the relative ordering of the levels, not on their numerical values. To break potential ties between subsets ordered by the LP with the same level, the time (subsets ordered earlier are treated as having lower levels) and then the size (larger subsets are treated as having lower levels) are used. This results in a complete order on the subsets ordered by the LP solution, captured by the \textit{index} of the subsets, defined below.

% \begin{definition}[Index of subsets]
% \label{def:index}
%    For any subset $S_{r}(t)$, where $t\in[T]$ and $r\in[n(t)]$, we define its \textit{$j$-index} $\alpha_j(S_r(t))\in\Nbb$ as the position of $S_r(t)$ among all subsets selected by the LP according to the lexicographic order on $(L_j(S_r(t)),t,r)$. That is, the order is determined first by increasing $j$-level, then (to break ties) by increasing time, and finally (for any remaining ties) by decreasing subset size.
% \end{definition}

% \begin{definition}[Index of subsets]
% \label{def:index}
%    For any subset $S=S_{r}(t)$, where $t\in[T]$ and $r\in[n(t)]$, we define its \textit{$j$-index} $\alpha_j(S,t,r)\in\Nbb$ as the position of $(S,t,r)$ among all subsets selected by the LP according to the lexicographic order on $(L_j(S),t,r)$. That is, the order is determined first by increasing $j$-level, then (to break ties) by increasing time, and finally (for any remaining ties) by decreasing subset size.
% \end{definition}

\begin{definition}[The $j$-index of subsets]
\label{def:index}
   For a subset $S$ ordered by the LP optimal solution at period $t$, the \textit{$j$-index} $\alpha_j(S,t)\in\Nbb$ is the position of $(S,t)$ among all subsets ordered by the LP solution according to the lexicographic order on $(L_j(S),t,-|S|)$. That is, the order is determined first by increasing $j$-level, then (to break ties) by increasing time, and finally (for any remaining ties) by decreasing subset size.
\end{definition}

\cref{lemma:monotonicity_levels} and the definition of the index in \cref{def:index} readily imply the following monotonicity property for the index.

\begin{corollary}\label{lemma:monotonicity_index}
    Let $j\in[k]$, $t\in[T]$, and $r_1<r_2\in[n(t)]$. Then, $\alpha_j(S_{r_1}(t),t) < \alpha_j(S_{r_2}(t),t)$.
\end{corollary}

One key insight regarding the indices is that any item $i\in P_j$ with demand interval $I$ can be associated with an index $\alpha$ such that the subsets ordered by the LP during $I$ that contain item $i$ are precisely those with $j$-index at most $\alpha$. The formal statement and proof of this property is deferred to \cref{lemma: consistent levels}.

% Another key property is that any item $i\in P_j$ with demand interval $I$ can be associated with a level $\ell$ such that the subsets ordered by the LP during $I$ that contain item $i$ are precisely those with $j$-level at most $\ell$. The formal statement requires further definitions and is deferred to \cref{lemma: consistent levels}.

Building on the notion of $j$-index, the algorithm can now be described.
The algorithm takes as input an optimal LP solution $(y_t^S)_{t\in[T],S\subseteq [N]}$ to \cref{eq: primal}, representing a schedule of fractional subset orders, and consists of two main steps. First, it constructs regions of the LP solution via a water-filling procedure and second, it places orders within each region. These two steps are carried out independently for each item category, after which the resulting orders across all categories are combined.

\subsection{Step 1: Construct regions using water-filling procedure. } \label{subsection: water filling} 
To conveniently describe the water-filling procedure, a continuous relaxation of the time periods $[T]$ is used. From now on, instead of viewing time discretely, each original time period $t$ is considered as a continuous interval $(t-1,t]$. In particular, the problem now spans the continuous time interval $(0,T]$. A subset $S$ ordered by the LP at time $t$ then corresponds to a subset ordered uniformly during the interval $(t-1,t]$. That is, $y_{s}^S=y_t^S$ for all $s\in(t-1,t]$ and all item subsets $S\subseteq [N]$. %Furthermore, we extend the total order and index of subsets to continuous time accordingly.

To illustrate the algorithm and its analysis, it is helpful to use a pictorial representation of the fractional solution in the 2-dimensional Cartesian plane (see \cref{fig: algorithm}a), where the \(x\)-axis denotes time and the \(y\)-axis denotes the fractional order quantity of various subsets. Fractional orders placed in each period are depicted as rectangles stacked on top of one another. This representation is used to demonstrate the water-filling procedure described below (see \cref{fig: algorithm}). In this coordinate system, each subset ordered by the LP is naturally associated with an \textit{area}. Specifically, for a subset $S$ ordered in period $s$, its area is given by the product of its fractional order quantity $y_s^S$ and the length of the time interval it spans.

\paragraph{Intuitive description of the water-filling procedure.}
The water-filling procedure for a given item category $j$ begins by filling the subset $S$ with the smallest $j$-index $\alpha_j(S)$, followed by the subset with second smallest $j$-index, and so on. Filling continues until the total filled area reaches $1/2$, with the last subset possibly only partially filled. The filled region constitutes the first \textit{layer}. The time interval $(0,T]$ is then partitioned into two intervals such that exactly half the filled area lies on each side of the split. Note that the split point may correspond to a non-integral time. 
The two filled regions on either side of time horizon are called \textit{level-sets}. They partition the first layer, with each level-set having area exactly $1/4$. An example of the resulting level-sets and the time-horizon split for the first layer are illustrated in \cref{fig: algorithm}b.

The same water-filling procedure is then applied independently to each half of the time horizon defined by the first layer's split. In each sub-interval, filling again begins with the subset of smallest index, regardless of whether it was filled in the previous layer, and continues until a total area of $1/2$ is filled, forming the second layer. Each sub-interval is then split into two parts, producing a partition of the second layer into four level-sets, each with area $1/4$. The four level-sets for the second layer and their corresponding time intervals determined by the splits are illustrated in green in \cref{fig: algorithm}c and \cref{fig: algorithm}d. 

This process continues recursively. If the water-filling procedure exhausts all available subsets within a given interval before reaching area $1/2$, then no level-set is created in that interval, and it is not split further. While this cannot occur in the first layer, it serves as the termination condition for subsequent layers. For example, in \cref{fig: algorithm}e, the third layer contains only four level-sets because the termination criterion is met for the right side of the time horizon.

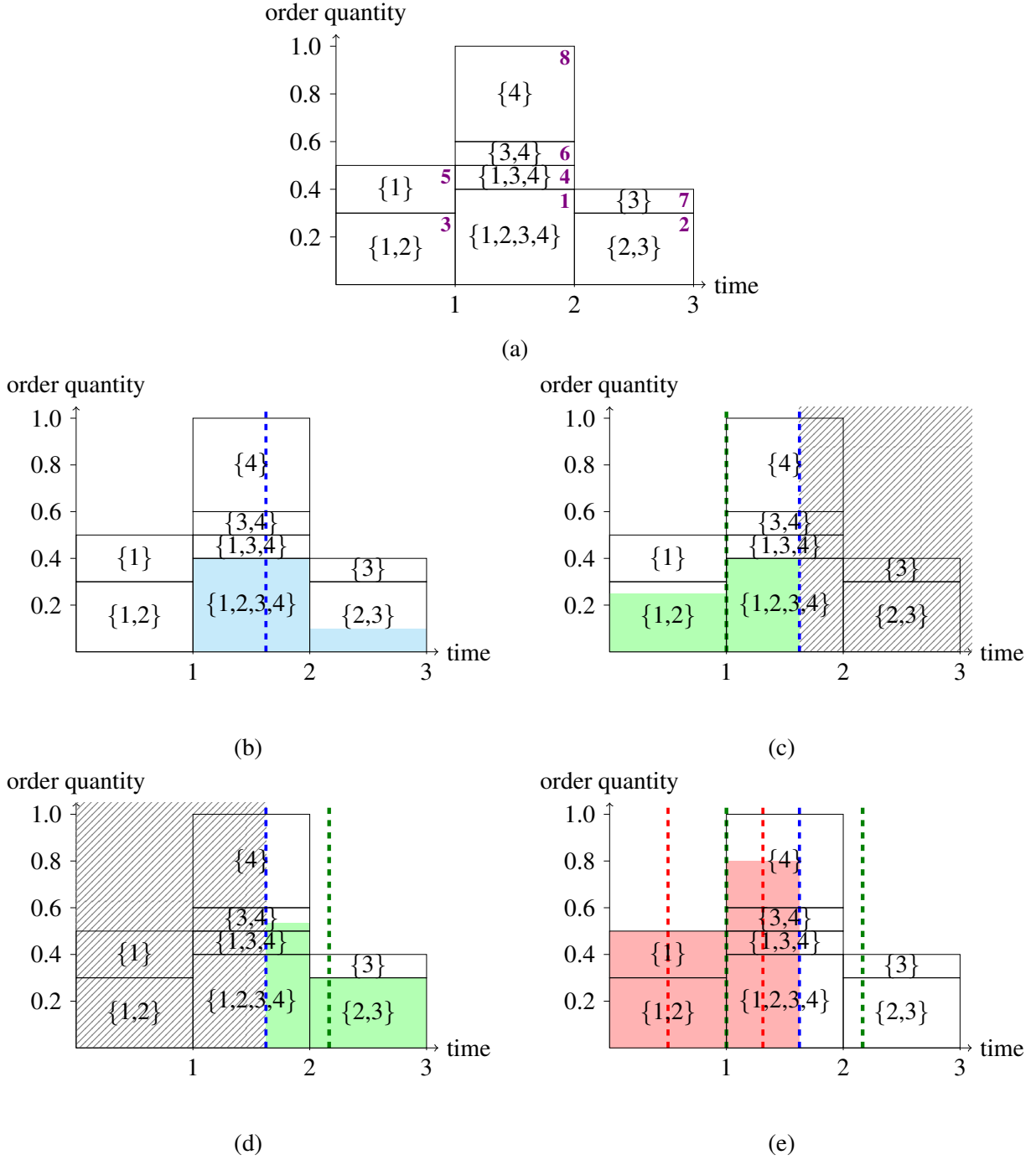
\begin{figure}[htbp]
    \centering

\begin{subfigure}{\textwidth}
\centering
\resizebox{.5\textwidth}{!}{
\begin{tikzpicture}[scale=2]

  % Axes
  \draw[->] (0,0) -- (3.1,0) node[right] {time};
  \draw[->] (0,0) -- (0,2.1) node[above] {order quantity};

  % Y-axis ticks and labels every 0.2, scaled by 2
  \foreach \y/\ylab in {0.4/0.2, 0.8/0.4, 1.2/0.6, 1.6/0.8, 2.0/1.0} {
    \draw (-0.05,\y) -- (0,\y);
    \node[left] at (-0.05,\y) {\ylab};
  }

  % X-axis ticks and labels at 1, 2, 3
  \foreach \x in {1,2,3} {
    \draw (\x,0) -- (\x,-0.02);
    \node[below] at (\x,-0.02) {\x};
  }

  % A tiny padding used to place the index in the top-right corner
  \def\pad{-0.04}

  % ----- Rectangles for x in [0,1] -----
  % {1,2}  -> index 3
  \draw[fill=none] (0,0) rectangle (1,0.6);
  \node at (0.5,0.30) {\{1,2\}};
  \node[font=\bfseries\footnotesize, text=violet, anchor=north east] at (1-\pad,0.6-\pad) {3};

  % {1}    -> index 5
  \draw[fill=none] (0,0.6) rectangle (1,1.0);
  \node at (0.5,0.80) {\{1\}};
  \node[font=\bfseries\footnotesize, text=violet, anchor=north east] at (1-\pad,1.0-\pad) {5};

  % ----- Rectangles for x in [1,2] -----
  % {1,2,3,4} -> index 1
  \draw[fill=none] (1,0) rectangle (2,0.8);
  \node at (1.5,0.40) {\{1,2,3,4\}};
  \node[font=\bfseries\footnotesize, text=violet, anchor=north east] at (2-\pad,0.8-\pad) {1};

  % {1,3,4}    -> index 4
  \draw[fill=none] (1,0.8) rectangle (2,1.0);
  \node at (1.5,0.90) {\{1,3,4\}};
  \node[font=\bfseries\footnotesize, text=violet, anchor=north east] at (2-\pad,1.0-\pad) {4};

  % {3,4}      -> index 6
  \draw[fill=none] (1,1.0) rectangle (2,1.2);
  \node at (1.5,1.10) {\{3,4\}};
  \node[font=\bfseries\footnotesize, text=violet, anchor=north east] at (2-\pad,1.2-\pad) {6};

  % {4}        -> index 8
  \draw[fill=none] (1,1.2) rectangle (2,2.0);
  \node at (1.5,1.60) {\{4\}};
  \node[font=\bfseries\footnotesize, text=violet, anchor=north east] at (2-\pad,2.0-\pad) {8};

  % ----- Rectangles for x in [2,3] -----
  % {2,3}  -> index 2
  \draw[fill=none] (2,0) rectangle (3,0.6);
  \node at (2.5,0.30) {\{2,3\}};
  \node[font=\bfseries\footnotesize, text=violet, anchor=north east] at (3-\pad,0.6-\pad) {2};

  % {3}    -> index 7
  \draw[fill=none] (2,0.6) rectangle (3,0.8);
  \node at (2.5,0.70) {\{3\}};
  \node[font=\bfseries\footnotesize, text=violet, anchor=north east] at (3-\pad,0.8-\pad) {7};

\end{tikzpicture}
}    
 \vspace{-0em}\caption{}
\end{subfigure}
    
    % Second row
    \begin{subfigure}{0.49\textwidth}
        \centering
        \resizebox{\linewidth}{!}{
\begin{tikzpicture}[scale=2]

  % Axes
  \draw[->] (0,0) -- (3.1,0) node[right] {time};
  \draw[->] (0,0) -- (0,2.1) node[above] {order quantity};

  % Y-axis ticks and labels every 0.2, scaled by 2
  \foreach \y/\ylab in {0.4/0.2, 0.8/0.4, 1.2/0.6, 1.6/0.8, 2.0/1.0} {
    \draw (-0.05,\y) -- (0,\y);
    \node[left] at (-0.05,\y) {\ylab};
  }

  % X-axis ticks and labels at 1, 2, 3
  \foreach \x in {1,2,3} {
    \draw (\x,0) -- (\x,-0.02);
    \node[below] at (\x,-0.02) {\x};
  }

  % === Shading (no border, behind everything else) ===
  \fill[cyan!20] (1,0) rectangle (2,0.8);   % Full shading
  \fill[cyan!20] (2,0) rectangle (3,0.2);   % Partial shading

  % === Vertical line that cuts shaded volume in half ===
  \draw[line width=1.5pt, blue, dashed] (1.625, 0) -- (1.625, 2.1);

  % === Box outlines and text ===
  \draw[fill=none] (0,0) rectangle (1,0.6);
  \node at (0.5,0.3) {\{1,2\}};

  \draw[fill=none] (0,0.6) rectangle (1,1.0);
  \node at (0.5,0.8) {\{1\}};

  \draw[fill=none] (1,0) rectangle (2,0.8);
  \node at (1.5,0.4) {\{1,2,3,4\}};

  \draw[fill=none] (1,0.8) rectangle (2,1.0);
  \node at (1.5,0.9) {\{1,3,4\}};

  \draw[fill=none] (1,1.0) rectangle (2,1.2);
  \node at (1.5,1.1) {\{3,4\}};

  \draw[fill=none] (1,1.2) rectangle (2,2.0);
  \node at (1.5,1.6) {\{4\}};

  \draw[fill=none] (2,0) rectangle (3,0.6);
  \node at (2.5,0.3) {\{2,3\}};

  \draw[fill=none] (2,0.6) rectangle (3,0.8);
  \node at (2.5,0.7) {\{3\}};

\end{tikzpicture}
        }
        \vspace{-0em}\caption{}
    \end{subfigure}
    \hfill
    \begin{subfigure}{0.49\textwidth}
        \centering
        \resizebox{\linewidth}{!}{
\begin{tikzpicture}[scale=2]

  % Axes
  \draw[->] (0,0) -- (3.1,0) node[right] {time};
  \draw[->] (0,0) -- (0,2.1) node[above] {order quantity};

  % === Gray overlay to the right of x = 1.625 ===
\fill[gray!20, pattern=north east lines, pattern color=gray] (1.625, 0) rectangle (3.1, 2.1);

  % Y-axis ticks and labels every 0.2, scaled by 2
  \foreach \y/\ylab in {0.4/0.2, 0.8/0.4, 1.2/0.6, 1.6/0.8, 2.0/1.0} {
    \draw (-0.05,\y) -- (0,\y);
    \node[left] at (-0.05,\y) {\ylab};
  }

  % X-axis ticks and labels at 1, 2, 3
  \foreach \x in {1,2,3} {
    \draw (\x,0) -- (\x,-0.02);
    \node[below] at (\x,-0.02) {\x};
  }

  % === Green shading ===
  \fill[green!30] (0,0) rectangle (1,0.5);                % Full box {1,2}
  \fill[green!30] (1,0) rectangle (1.625,0.8);            % Left part of {1,2,3,4}

  % === Vertical lines
  \draw[line width=1.5pt, blue, dashed] (1.625, 0) -- (1.625, 2.1);
  \draw[line width=1.8pt, draw=green!50!black, dashed] (1, 0) -- (1, 2.1);

  % === Box outlines and text ===
  \draw[fill=none] (0,0) rectangle (1,0.6);
  \node at (0.5,0.3) {\{1,2\}};

  \draw[fill=none] (0,0.6) rectangle (1,1.0);
  \node at (0.5,0.8) {\{1\}};

  \draw[fill=none] (1,0) rectangle (2,0.8);
  \node at (1.5,0.4) {\{1,2,3,4\}};

  \draw[fill=none] (1,0.8) rectangle (2,1.0);
  \node at (1.5,0.9) {\{1,3,4\}};

  \draw[fill=none] (1,1.0) rectangle (2,1.2);
  \node at (1.5,1.1) {\{3,4\}};

  \draw[fill=none] (1,1.2) rectangle (2,2.0);
  \node at (1.5,1.6) {\{4\}};

  \draw[fill=none] (2,0) rectangle (3,0.6);
  \node at (2.5,0.3) {\{2,3\}};

  \draw[fill=none] (2,0.6) rectangle (3,0.8);
  \node at (2.5,0.7) {\{3\}};
\end{tikzpicture}
        }
        \vspace{-0em}\caption{}
    \end{subfigure}

    \vspace{-.1em}

    % Third row
    \begin{subfigure}{0.49\textwidth}
        \centering
        \resizebox{\linewidth}{!}{
\begin{tikzpicture}[scale=2]

  % Axes
  \draw[->] (0,0) -- (3.1,0) node[right] {time};
  \draw[->] (0,0) -- (0,2.1) node[above] {order quantity};

  % === Gray overlay to the LEFT of x = 1.625 ===
  \fill[gray!20, pattern=north east lines, pattern color=gray] (0, 0) rectangle (1.625, 2.1);

  % Y-axis ticks and labels every 0.2, scaled by 2
  \foreach \y/\ylab in {0.4/0.2, 0.8/0.4, 1.2/0.6, 1.6/0.8, 2.0/1.0} {
    \draw (-0.05,\y) -- (0,\y);
    \node[left] at (-0.05,\y) {\ylab};
  }

  % X-axis ticks and labels at 1, 2, 3
  \foreach \x in {1,2,3} {
    \draw (\x,0) -- (\x,-0.02);
    \node[below] at (\x,-0.02) {\x};
  }

  % === Green shading (now to the RIGHT of blue line) ===
  \fill[green!30] (1.625,0) rectangle (2,1.06666);            % Right part of {1,2,3,4}
  \fill[green!30] (2,0) rectangle (3,0.6);                %  {2,3}

  % === Vertical lines
  \draw[line width=1.5pt, blue, dashed] (1.625, 0) -- (1.625, 2.1);
  \draw[line width=1.8pt, draw=green!50!black, dashed] (2.166, 0) -- (2.166, 2.1);

  % === Box outlines and text ===
  \draw[fill=none] (0,0) rectangle (1,0.6);
  \node at (0.5,0.3) {\{1,2\}};

  \draw[fill=none] (0,0.6) rectangle (1,1.0);
  \node at (0.5,0.8) {\{1\}};

  \draw[fill=none] (1,0) rectangle (2,0.8);
  \node at (1.5,0.4) {\{1,2,3,4\}};

  \draw[fill=none] (1,0.8) rectangle (2,1.0);
  \node at (1.5,0.9) {\{1,3,4\}};

  \draw[fill=none] (1,1.0) rectangle (2,1.2);
  \node at (1.5,1.1) {\{3,4\}};

  \draw[fill=none] (1,1.2) rectangle (2,2.0);
  \node at (1.5,1.6) {\{4\}};

  \draw[fill=none] (2,0) rectangle (3,0.6);
  \node at (2.5,0.3) {\{2,3\}};

  \draw[fill=none] (2,0.6) rectangle (3,0.8);
  \node at (2.5,0.7) {\{3\}};
\end{tikzpicture}
        }
        \vspace{-0em}\caption{}
    \end{subfigure}
    \hfill
    \begin{subfigure}{0.49\textwidth}
        \centering
        \resizebox{\linewidth}{!}{
\begin{tikzpicture}[scale=2]

  % Axes
  \draw[->] (0,0) -- (3.1,0) node[right] {time};
  \draw[->] (0,0) -- (0,2.1) node[above] {order quantity};

  % % === Gray overlay to the LEFT of x = 1.625 ===
  % \fill[gray!20, pattern=north east lines, pattern color=gray] (0, 0) rectangle (1.625, 2.1);

  % Y-axis ticks and labels every 0.2, scaled by 2
  \foreach \y/\ylab in {0.4/0.2, 0.8/0.4, 1.2/0.6, 1.6/0.8, 2.0/1.0} {
    \draw (-0.05,\y) -- (0,\y);
    \node[left] at (-0.05,\y) {\ylab};
  }

  % X-axis ticks and labels at 1, 2, 3
  \foreach \x in {1,2,3} {
    \draw (\x,0) -- (\x,-0.02);
    \node[below] at (\x,-0.02) {\x};
  }

  % === Red shading ===
  \fill[red!30] (1,0) rectangle (1.625,1.6);            % Left part of {1,2,3,4}
    \fill[red!30] (0,0) rectangle (1,1);            % Left part of {1,2,3,4}

  % === Vertical lines
  \draw[line width=1.5pt, blue, dashed] (1.625, 0) -- (1.625, 2.1);
  \draw[line width=1.8pt, draw=green!50!black, dashed] (2.166, 0) -- (2.166, 2.1);
\draw[line width=1.8pt, draw=green!50!black, dashed] (1, 0) -- (1, 2.1);
  \draw[line width=1.5pt, red, dashed] (1+5/16, 0) -- (1+5/16, 2.1);
  \draw[line width=1.5pt, red, dashed] (.5, 0) -- (.5, 2.1);

  % === Box outlines and text ===
  \draw[fill=none] (0,0) rectangle (1,0.6);
  \node at (0.5,0.3) {\{1,2\}};

  \draw[fill=none] (0,0.6) rectangle (1,1.0);
  \node at (0.5,0.8) {\{1\}};

  \draw[fill=none] (1,0) rectangle (2,0.8);
  \node at (1.5,0.4) {\{1,2,3,4\}};

  \draw[fill=none] (1,0.8) rectangle (2,1.0);
  \node at (1.5,0.9) {\{1,3,4\}};

  \draw[fill=none] (1,1.0) rectangle (2,1.2);
  \node at (1.5,1.1) {\{3,4\}};

  \draw[fill=none] (1,1.2) rectangle (2,2.0);
  \node at (1.5,1.6) {\{4\}};

  \draw[fill=none] (2,0) rectangle (3,0.6);
  \node at (2.5,0.3) {\{2,3\}};

  \draw[fill=none] (2,0.6) rectangle (3,0.8);
  \node at (2.5,0.7) {\{3\}};
\end{tikzpicture}
        }
        \vspace{-0em}\caption{}
    \end{subfigure}

\caption{
This figure illustrates Step~1 of Algorithm~\ref{alg:water_filling} on an instance with \( N = 4 \) and \( T = 3 \). Subfigure (a) shows the fractional solution to \eqref{eq: primal} with the $j$-index of each subset indicated in purple. In (b), subsets with the lowest indices are filled to a total area of \(1/2\); the blue region represents the first layer, and the dotted blue line splits the time horizon so that each side contains exactly half of the filled area. In (c) and (d), the same procedure is applied recursively to each sub-interval, producing the second layer (green). In (e), only two of the four resulting sub-intervals has area at least \(1/2\); they are filled to \(1/2\) (red, indicating the third layer) and again split evenly (dotted red line). All intervals now have area below \(1/2\), satisfying the termination condition.
}
\label{fig: algorithm}
\end{figure}

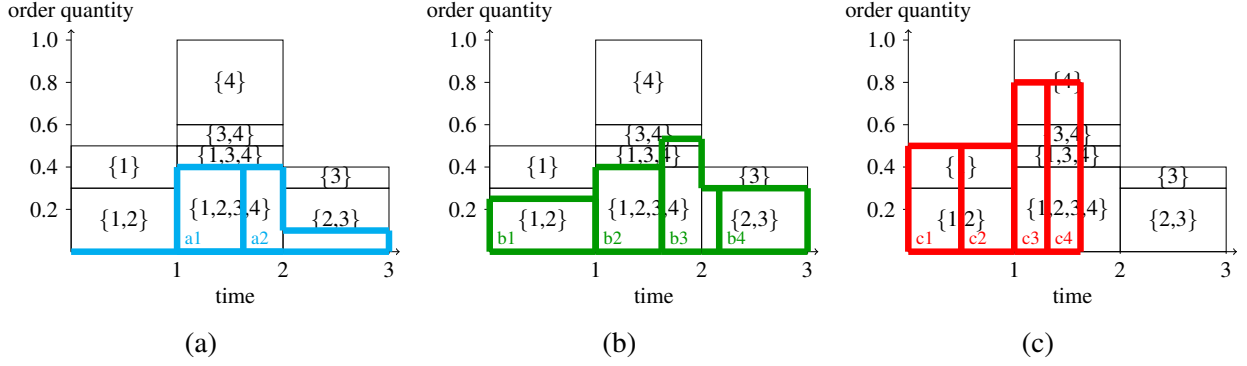
\begin{figure}[tbp]
       \begin{subfigure}{0.33\textwidth}
        \centering
        \resizebox{\linewidth}{!}{
\begin{tikzpicture}[scale=2]

  % Axes
  \draw[->] (0,0) -- (3.1,0) node[midway, below=15pt] {time};
  \draw[->] (0,0) -- (0,2.1) node[above] {order quantity};

  % Y-axis ticks and labels every 0.2, scaled by 2
  \foreach \y/\ylab in {0.4/0.2, 0.8/0.4, 1.2/0.6, 1.6/0.8, 2.0/1.0} {
    \draw (-0.05,\y) -- (0,\y);
    \node[left] at (-0.05,\y) {\ylab};
  }

  % X-axis ticks and labels at 1, 2, 3
  \foreach \x in {1,2,3} {
    \draw (\x,0) -- (\x,-0.02);
    \node[below] at (\x,-0.02) {\x};
  }

  % Rectangles for x in [0,1], y scaled
  \draw[fill=none] (0,0) rectangle (1,0.6);
  \node at (0.5,0.3) {\{1,2\}};

  \draw[fill=none] (0,0.6) rectangle (1,1.0);
  \node at (0.5,0.8) {\{1\}};

  % Rectangles for x in [1,2], y scaled
  \draw[fill=none] (1,0) rectangle (2,0.8);
  \node at (1.5,0.4) {\{1,2,3,4\}};

  \draw[fill=none] (1,0.8) rectangle (2,1.0);
  \node at (1.5,0.9) {\{1,3,4\}};

  \draw[fill=none] (1,1.0) rectangle (2,1.2);
  \node at (1.5,1.1) {\{3,4\}};

  \draw[fill=none] (1,1.2) rectangle (2,2.0);
  \node at (1.5,1.6) {\{4\}};

  % Rectangles for x in [2,3], y scaled
  \draw[fill=none] (2,0) rectangle (3,0.6);
  \node at (2.5,0.3) {\{2,3\}};

  \draw[fill=none] (2,0.6) rectangle (3,.8);
  \node at (2.5,0.7) {\{3\}};

       % Cyan box 1
       \draw[cyan, line width=3.5pt] (0,0) -- (2,0);
        \draw[cyan, line width=3.5pt] (1,0) -- (2,0);
        \draw[cyan, line width=3.5pt] (1,0.8) -- (2,0.8);
        \draw[cyan, line width=3.5pt] (1,0) -- (1,0.8);
        \draw[cyan, line width=3.5pt] (2,.2) -- (2,0.8);

        % Cyan box 2
        \draw[cyan, line width=3.5pt] (2,0) -- (3,0);
        \draw[cyan, line width=3.5pt] (2,0.2) -- (3,0.2);
        % \draw[cyan, line width=3.5pt] (2,0) -- (2,0.2);
        \draw[cyan, line width=3.5pt] (3,0) -- (3,0.2);

        % Vertical split inside first cyan box
        \draw[cyan, line width=3.5pt] (1.625, 0) -- (1.625, 0.8);

  \node[cyan, font=\footnotesize] at (1,0) [above right] {a1};
  \node[cyan, font=\footnotesize] at (1.625,0) [above right] {a2};
\end{tikzpicture}
        }
        \caption{}
    \end{subfigure}
       \begin{subfigure}{0.33\textwidth}
        \centering
        \resizebox{\linewidth}{!}{
\begin{tikzpicture}[scale=2]

  % Axes
  \draw[->] (0,0) -- (3.1,0) node[midway, below=15pt] {time};
  \draw[->] (0,0) -- (0,2.1) node[above] {order quantity};

  % Y-axis ticks and labels every 0.2, scaled by 2
  \foreach \y/\ylab in {0.4/0.2, 0.8/0.4, 1.2/0.6, 1.6/0.8, 2.0/1.0} {
    \draw (-0.05,\y) -- (0,\y);
    \node[left] at (-0.05,\y) {\ylab};
  }

  % X-axis ticks and labels at 1, 2, 3
  \foreach \x in {1,2,3} {
    \draw (\x,0) -- (\x,-0.02);
    \node[below] at (\x,-0.02) {\x};
  }

  % Rectangles for x in [0,1], y scaled
  \draw[fill=none] (0,0) rectangle (1,0.6);
  \node at (0.5,0.3) {\{1,2\}};

  \draw[fill=none] (0,0.6) rectangle (1,1.0);
  \node at (0.5,0.8) {\{1\}};

  % Rectangles for x in [1,2], y scaled
  \draw[fill=none] (1,0) rectangle (2,0.8);
  \node at (1.5,0.4) {\{1,2,3,4\}};

  \draw[fill=none] (1,0.8) rectangle (2,1.0);
  \node at (1.5,0.9) {\{1,3,4\}};

  \draw[fill=none] (1,1.0) rectangle (2,1.2);
  \node at (1.5,1.1) {\{3,4\}};

  \draw[fill=none] (1,1.2) rectangle (2,2.0);
  \node at (1.5,1.6) {\{4\}};

  % Rectangles for x in [2,3], y scaled
  \draw[fill=none] (2,0) rectangle (3,0.6);
  \node at (2.5,0.3) {\{2,3\}};

  \draw[fill=none] (2,0.6) rectangle (3,.8);
  \node at (2.5,0.7) {\{3\}};

        % Green boxes (manually drawn)
        \draw[green!60!black, line width=3.5pt] (0,0) -- (1,0);
        \draw[green!60!black, line width=3.5pt] (0,0.5) -- (1,0.5);
        \draw[green!60!black, line width=3.5pt] (0,0) -- (0,0.5);
        \draw[green!60!black, line width=3.5pt] (1,0) -- (1,0.5);

        \draw[green!60!black, line width=3.5pt] (1,0) -- (1.625,0);
        \draw[green!60!black, line width=3.5pt] (1,0.8) -- (1.625,0.8);
        \draw[green!60!black, line width=3.5pt] (1,0) -- (1,0.8);
        \draw[green!60!black, line width=3.5pt] (1.625,0) -- (1.625,0.8);

        \draw[green!60!black, line width=3.5pt] (1.625,0) -- (2,0);
        \draw[green!60!black, line width=3.5pt] (1.625,1.06666) -- (2,1.06666);
        \draw[green!60!black, line width=3.5pt] (1.625,0) -- (1.625,1.06666);
        \draw[green!60!black, line width=3.5pt] (2,0.6) -- (2,1.06666);

        \draw[green!60!black, line width=3.5pt] (2,0) -- (3,0);
        \draw[green!60!black, line width=3.5pt] (2,0.6) -- (3,0.6);
        % \draw[green!60!black, line width=3.5pt] (2,0) -- (2,0.6);
        \draw[green!60!black, line width=3.5pt] (3,0) -- (3,0.6);

        % Internal vertical split
        \draw[green!60!black, line width=3.5pt] (2.166, 0) -- (2.166, 0.6);

  \node[green!60!black, font=\footnotesize] at (0,0) [above right] {b1};
  \node[green!60!black, font=\footnotesize] at (1,0) [above right] {b2};
  \node[green!60!black, font=\footnotesize] at (1.625,0) [above right] {b3};
  \node[green!60!black, font=\footnotesize] at (2.166,0) [above right] {b4};

\end{tikzpicture}
        }
       \caption{}
    \end{subfigure}
       \begin{subfigure}{0.33\textwidth}
        \centering
        \resizebox{\linewidth}{!}{
\begin{tikzpicture}[scale=2]

  % Axes
  \draw[->] (0,0) -- (3.1,0) node[midway, below=15pt] {time};
  \draw[->] (0,0) -- (0,2.1) node[above] {order quantity};

  % Y-axis ticks and labels every 0.2, scaled by 2
  \foreach \y/\ylab in {0.4/0.2, 0.8/0.4, 1.2/0.6, 1.6/0.8, 2.0/1.0} {
    \draw (-0.05,\y) -- (0,\y);
    \node[left] at (-0.05,\y) {\ylab};
  }

  % X-axis ticks and labels at 1, 2, 3
  \foreach \x in {1,2,3} {
    \draw (\x,0) -- (\x,-0.02);
    \node[below] at (\x,-0.02) {\x};
  }

  % Rectangles for x in [0,1], y scaled
  \draw[fill=none] (0,0) rectangle (1,0.6);
  \node at (0.5,0.3) {\{1,2\}};

  \draw[fill=none] (0,0.6) rectangle (1,1.0);
  \node at (0.5,0.8) {\{1\}};

  % Rectangles for x in [1,2], y scaled
  \draw[fill=none] (1,0) rectangle (2,0.8);
  \node at (1.5,0.4) {\{1,2,3,4\}};

  \draw[fill=none] (1,0.8) rectangle (2,1.0);
  \node at (1.5,0.9) {\{1,3,4\}};

  \draw[fill=none] (1,1.0) rectangle (2,1.2);
  \node at (1.5,1.1) {\{3,4\}};

  \draw[fill=none] (1,1.2) rectangle (2,2.0);
  \node at (1.5,1.6) {\{4\}};

  % Rectangles for x in [2,3], y scaled
  \draw[fill=none] (2,0) rectangle (3,0.6);
  \node at (2.5,0.3) {\{2,3\}};

  \draw[fill=none] (2,0.6) rectangle (3,.8);
  \node at (2.5,0.7) {\{3\}};

        % Red region (fully drawn manually)
        \draw[red, line width=3.5pt] (1,0) -- (1.625,0);           % bottom
        \draw[red, line width=3.5pt] (1,1.6) -- (1.625,1.6);       % top
        \draw[red, line width=3.5pt] (1,0) -- (1,1.6);             % left
        \draw[red, line width=3.5pt] (1.625,0) -- (1.625,1.6);     % right

                % Red region (fully drawn manually)
        \draw[red, line width=3.5pt] (0,0) -- (1,0);           % bottom
        \draw[red, line width=3.5pt] (0,0) -- (0,1);             % left
        \draw[red, line width=3.5pt] (0,1) -- (1,1);   

        % Vertical split inside
        \draw[red, line width=3.5pt] (1+5/16, 0) -- (1+5/16, 1.6);
        \draw[red, line width=3.5pt] (.5, 0) -- (.5, 1);

  \node[red, font=\footnotesize] at (0,0) [above right] {c1};
  \node[red, font=\footnotesize] at (0.5,0) [above right] {c2};
  \node[red, font=\footnotesize] at (1,0) [above right] {c3};
  \node[red, font=\footnotesize] at (1+5/16,0) [above right] {c4};

\end{tikzpicture}
        }
        \caption{}
    \end{subfigure}

\caption{The level-sets constructed in Layers 1, 2, and 3, shown from left to right. Layer 1 contains two level-sets (shown in blue), Layer 2 contains four (shown in green), and Layer 3 contains four (shown in red). 
% While some regions may not appear rectangular in this visualization, this is due to the y-axis representing order quantity, rather than index. 
Note that although level-set a1 does not include any subset in the time interval $(0,1]$, it still spans that interval.}
    \label{fig:constructed regions}
\end{figure}

In summary, for each category \(j\in[k]\), the water-filling procedure produces a sequence of layers, each of which is partitioned into disjoint level-sets.
An example of the full procedure is shown in \cref{fig: algorithm}, and the resulting layers and level-sets within each layer are illustrated in \cref{fig:constructed regions}. 

\paragraph{Formal description of the water-filling procedure.}
The goal of the water-filling procedure is to construct a collection of \textit{level-sets}. The level-sets are in different layers, and each level-set is fully characterized by three parameters: a time interval \(I\), an index \(\alpha\), and a fractional quantity \(\beta \in (0,1]\). Specifically, the level-set contains all subsets within the interval \(I\) whose index is at most \(\alpha\), with the subset of index \(\alpha\) included only up to a proportion $\beta$. \cref{def:level-set} provides a precise mathematical definition of level-sets:
%we allow $\alpha$ to take continuous, rather than integer values, where the fractional part reflects the proportion of the final subset included. 
%In other words, let an index of $\alpha$ indicate that the subset with index $\floor{\alpha}+1$ is only partially filled, up to a proportion $\alpha-\floor{\alpha}$. 

\begin{definition}[Level-set]\label{def:level-set}
Fix a category $j\in[k]$, a continuous time interval $I$, an index $\alpha\in\mathbb N$, and a proportion $\beta\in(0,1]$. 
Consider the collection of all subsets ordered by the optimal LP solution during time interval $I$ with $j$-index at most $\alpha$, that is, $$\set{(S,t): t\in I,\,r\in[n(t)],\, S=S_r(t),\, \alpha_j(S,t)\leq \alpha}.$$ 
The level-set $\Rect_j(\alpha,\beta;I)$ consists of these fractional LP orders, except the subset with $j$-index equal to $\alpha$ is only included up to a proportion $\beta$.
\end{definition}

Recall that each subset ordered by the LP solution is associated with an area in the order quantity–time coordinate system, defined as the product of its order quantity and the length of the time interval it spans. Level-sets admit an analogous notion of area, i.e., the total area of all subsets they contain, where partially included subsets contribute proportionally to the extent of their inclusion.
Specifically, for a level-set $\Rect_j(\alpha,\beta;I)$, its area, denoted $\Area_j  \paren{\Rect_j(\alpha,\beta;I)}$, is the total area over the interval $I$ of all subsets with $j$-index at most $\alpha$, where the subset with $j$-index $\alpha$ has area contribution that is scaled by $\beta$. 
For convenience, let $\Area_j(I):=\Area_j(\Rect_j(\infty,1; I))$  denote the area of all subsets ordered by the LP solution in the time interval $I$.

% \begin{definition}[Level-set]\label{def:level-set}
%     Fix a category $j\in[k]$, a time interval $I$, and an index $\alpha\geq 0$. The level-set $\Rect_j(\alpha;I)$ is the region containing all subsets ordered by the LP solution during interval $I$ with $j$-index at most $\floor{\alpha}$, together with $\alpha-\floor{\alpha}$ proportion of the level-set with index $\floor{\alpha}+1$. Formally, the level-set is defined as $\Rect_j(\alpha;I):=\set{(S,t): t\in I,\,r\in[n(t)],\,S=S_r(t),\,q_r(t;\alpha)>0}$, where 
%     \begin{equation*}
%         %\Rect_j(\alpha;I) := \set{ (S_r(t),q_r(t;\alpha)): t\in I, r\in[n(t)] } \text{ where } 
%         q_r(t;\alpha):=\begin{cases}
%             1 & \alpha_j(S_r(t)) \leq \floor{\alpha}\\
%             \alpha-\floor{\alpha} & \alpha_j(S_r(t)) = \floor{\alpha} +1 \\
%             0 & \alpha_j(S_r(t)) > \floor{\alpha} +1,
%         \end{cases}
%     \end{equation*}
%     denotes the proportion of subset $S_r(t)$ included in the level-set.
% The volume of $\Rect_j(\alpha;I)$ is defined as the total order quantity of the included subsets:
%     \begin{equation*}
%         \Area_j \paren{\Rect_j(\alpha;I)}:= \int_{t\in I} \paren{\sum_{r\in[n(t)]} y_r(t) q_r(t;\alpha)} dt.
%     \end{equation*}
% \end{definition}

With this formal notion of level-sets and their associated area, the water-filling algorithm can be described as follows. The algorithm is applied recursively to working time intervals $I$, starting with the full time horizon $(0,T]$. For a given time interval $I:=(t_1,t_2]$, the first step is to identify the level-set on $I$ with area $1/2$; that is, find the $\alpha\in\mathbb N$ and $\beta\in(0,1]$ such that $\Area_j(\Rect_j(\alpha,\beta;I))=1/2$ (if no such level-set exists, the procedure for interval $I$ terminates). Next, determine a cutoff time $t\in I$ that splits this level-set in two parts of equal area, i.e., such that $\Area_j(\Rect_j(\alpha,\beta;(t_1,t]))=1/4$. This yields the two regions $\Rect_j(\alpha,\beta;(t_1,t])$ and $\Rect_j(\alpha,\beta;(t,t_2])$, which make up the first \textit{layer} for category $j$. The same procedure is then applied to the sub-intervals $(t_1,t]$ and $(t,t_2]$, generating additional regions in the next layer. This water-filling procedure is summarized in Step 1 of Algorithm~\ref{alg:water_filling}.

\paragraph{Tree structure of level-sets.}
This recursive construction naturally induces a tree structure. For each category $j$, consider a rooted tree whose nodes correspond to level-sets. The root node is given by the dummy level-set $\{\Rect_j(0,1;(0,T])\}$ that encompasses the entire time horizon. If applying the procedure to a level-set associated with an interval $I=(t_1,t_2]$ yields a split at time $t$, then the corresponding node has two children, given by the level-sets defined on the sub-intervals $(t_1,t]$ and $(t,t_2]$. Taking the root to be at depth $0$, the nodes at depth $\ell$ in the tree correspond exactly to the level-sets in the $\ell$-th layer, denoted $\Lcal_\ell(j)$. Let $\Lcal(j) := \bigcup_{\ell \geq 0} \Lcal_\ell(j)$ denote the set of all level-sets across layers. When the category $j$ is clear from context, the dependence on $j$ is omitted.
For any level-set $R \in \Lcal$, let $\mathrm{Dsc}(R)$ denote the set of its descendants, i.e., all level-sets in $\Lcal$ for which $R$ is an ancestor under this tree structure, and let $\mathrm{Ch}(R)$ denote the children of $R$.

% within each layer, the level-sets correspond to disjoint time intervals, and these intervals are nested across layers by construction. Specifically, any level-set $\Rect_j(\alpha,\beta;I)$ in layer $\Lcal_l$ for $l\geq 1$ is formed by splitting into two the time interval $I_p$ of some level-set $\Rect_j(\alpha_p,\beta_p;I_p)\in\Lcal_{l-1}$ in the previous layer (see Algorithm \ref{alg:water_filling}). We refer to $\Rect_j(\alpha_p,\beta_p;I_p)$ as the \emph{parent} of $\Rect_j(\alpha,\beta;I)$. For example, in \cref{fig:constructed regions}, level-set a1 is the parent of level-sets b1 and b2, while a2 is the parent of b3 and b4. Similarly, b1 is the parent of c1 and c2, and b2 is the parent of c3 and c4. 

\subsection{Step 2: Placing orders to satisfy demands for items in $P_j$. } \label{subsection: placing orders}

Step 2 of Algorithm~\ref{alg:water_filling} aims to select and order one subset from each level-set in $\Lcal$ while minimizing the total ordering cost. For convenience, we say that a level-set is \emph{satisfied} if it contains an element $(S,t)$ such that subset $S$ is ordered at time $t$. Crucially, an order placed at a node $R = \Rect_j(\alpha,\beta;I)$ satisfies all nodes along a single root-to-leaf path in the subtree rooted at $R$. Specifically, an order at time $t \in I$ corresponds to a unique root-to-leaf path consisting of descendant nodes whose associated intervals contain $t$, and it satisfies exactly those descendants. Consequently, once an order is placed at node $R$, no additional orders are required for the nodes along one path in its subtree. This will be important when bounding the cost of the orders.

\begin{algorithm}[t]

\caption{Algorithm for satisfying demands associated with items in category $j$}\label{alg:water_filling}

%\SetAlgoLined
\LinesNumbered
\everypar={\nl}

% \hrule height\algoheightrule\kern3pt\relax
\KwIn{Category $j\in[k]$. LP solution: $S_r(t),y_r(t)$ for $t\in(0,T], r\in[n(t)]$}
\KwOut{Schedule of orders $\Ocal^{(j)}$ for category $j$}

% \vspace{3mm}

{\nonl \textbf{Step 1: Construction of regions for each layer}}

Initialize $\Ical_1:=\set{ (0,T] }$ and $l\gets 1$

\While{$\Ical_l \neq \emptyset$}{
    Initialize layer $\Lcal_l\gets \emptyset$ and intervals for the next layer $\Ical_{l+1}\gets \emptyset$
    
    \For{$I=(t_1,t_2]\in \Ical_l$ such that $\Area_j(I) \geq 1/2$}{
            Find the $\alpha\in\mathbb N$ and $\beta\in(0,1]$ with $\Area_j(\Rect_j(\alpha,\beta;I)) = 1/2$

            Find smallest $t\in I$ such that $\Area_j(\Rect_j(\alpha,\beta; (t_1,t]))=1/4$

            Update the layer $\Lcal_l\gets \Lcal_l \cup \{\Rect_j(\alpha,\beta;(t_1,t]), \Rect_j(\alpha,\beta;(t,t_2]) \}$

            Update intervals for the next layer $\Ical_{l+1} \gets \Ical_{l+1} \cup\set{(t_1,t],(t,t_2]}$
    }
    $l\gets l+1$
}

% \vspace{3mm}
{\nonl \textbf{Step 2: Order placements for each level-set}}

Solve the dynamic program $\sum_{R\in\Lcal_1}\DP(R).$

Set $\Ocal^{(j)}$ to the set of pairs $(S, t)$ that attain the minimum dynamic program solution%; that is, they satisfy $\Rect_j(0,0; (0, T])$ and all its descendants with minimal cost.

\Return $\Ocal^{(j)}$

% \hrule height\algoheightrule\kern3pt\relax
\end{algorithm}

To deterministically select subsets that satisfy all nodes in $\Lcal$ with minimum cost, we formulate a dynamic programming (DP) procedure over the tree induced by the level-sets. %Recall that each node corresponds to a level-set $R \in \Lcal$, and that placing an order at time $t$ in $R$ satisfies exactly the nodes along a single root-to-leaf path in the subtree rooted at $R$.
Each DP state corresponds to a node $R \in \Lcal$. The value $\DP(R)$ denotes the minimum cost required to satisfy $R$ and all of its descendants. To compute $\DP(R)$, consider all candidate orders $(S,t) \in R$. Placing such an order incurs cost $K(S)$ and satisfies $R$, as well as all descendants whose intervals contain $t$, i.e., those lying along a single path in the subtree rooted at $R$.

A key subtlety is how to account for the remaining nodes that are \emph{not} satisfied by this order. These correspond precisely to the subtrees branching off from this path. In particular, we sum over the roots of the maximal unsatisfied subtrees, i.e., the children of nodes along the path that are not themselves on the path.
For a fixed time $t$, define the residual cost $H(R,t)$ as the minimum cost required to satisfy all descendants of $R$ that are not already satisfied by an order at time $t$. This can be computed recursively as
\begin{align*}
H(R,t)
=
\sum_{C \in \mathrm{Ch}(R)}
\begin{cases}
H(C,t), & \text{if } t \in I_C, \\
\DP(C), & \text{if } t \notin I_C,
\end{cases}
\end{align*}
where $I_C$ denotes the time interval associated with node $C$. If $R$ is a leaf, then $H(R,t)=0$.

The DP recursion is then given by
\begin{align}\label{eq: DP}
\DP(R)
=
\min_{(S,t)\in R}
\left\{
K(S) + H(R,t)
\right\}.
\end{align}

Intuitively, once an order $(S,t)$ is placed at node $R$, it ``covers'' a single path in the subtree. Along this path, no further cost is needed for the nodes themselves, but we must still account for any side branches that are not covered. 

% The final policy is obtained from minimizing the cost of satisfying all regions $\sum_{R\in\Lcal_1}\DP(R)$ by keeping track of the orders $(S,t)$ that attain the minimum in \eqref{eq: DP}. The resulting set of orders $\Ocal^{(j)}$ satisfies all level-sets in $\Lcal$ with minimum total cost. This procedure is summarized in Step 2 of Algorithm~\ref{alg:water_filling}, and it is verified in \cref{lemma: runtime} that it runs in polynomial time. In \cref{subsection: feasibility of alg for decomposition}, we show that these orders fulfill all demand associated with items in $P_j$. 

The final schedule is obtained by minimizing the total cost of satisfying all first-layer level-sets, $\sum_{R\in\Lcal_1}\DP(R),$ while keeping track of the orders $(S,t)$ that attain the minimum in \eqref{eq: DP}. The resulting set of orders $\Ocal^{(j)}$ is a minimum-cost collection of orders that satisfies all level-sets in $\Lcal$. This procedure is summarized in Step~2 of Algorithm~\ref{alg:water_filling}, and \cref{lemma: runtime} shows that it can be implemented in polynomial time. In \cref{subsection: feasibility of alg for decomposition}, we prove that these orders satisfy all demand associated with items in $P_j$.

% To deterministically select subsets that satisfy all nodes in $\Lcal$ with minimum cost, we formulate a dynamic programming (DP) procedure over the tree, which exploits this tree structure. Each DP state corresponds to a single level-set $R \in \Lcal$. The value function $\text{DP}(R)$ denotes the minimum cost needed to satisfy $R$, as well as all of its descendants in $\Lcal$. To evaluate $\text{DP}(R)$, we consider all possible subset–time pairs $(S,t)$ contained in $R$—that is, all candidate orders that could satisfy $R$. Placing an order $(S,t)$ incurs an immediate cost $K(S)$, and may also help satisfy some of $R$’s descendants if $t$ lies within their time intervals. For each descendant $\Rect_j(\alpha',\beta';I') \in \text{Dsc}(R)$, if $t \notin I'$ (meaning the order at $t$ does not satisfy that descendant), we must pay the cost of satisfying it separately, given by $\text{DP}(\Rect_j(\alpha',\beta';I'))$. This leads to the recursion:
% \begin{align}\label{eq: DP}
% \begin{split}
%     &\text{DP}(R)= \min_{(S,t)\in R}\bigg\{K(S) \\ 
%     &\qquad\qquad +\sum_{\substack{\Rect_j(\alpha',\beta';I')\\\in\text{Dsc}(R)}}\text{DP}(\Rect_j(\alpha',\beta';I')) \1[t\notin I'] \bigg\}.
% \end{split}
% \end{align}

% Intuitively, the DP chooses the order $(S,t)$ from $R$ that minimizes the total cost of placing that order, plus the additional costs needed to satisfy any descendants not already satisfied by it. If $R$ has no descendants, the summation term is zero, so the cost is simply the cheapest way to satisfy $R$ itself.

\begin{algorithm}[h]

\caption{Approximation Algorithm for the DICP} 

\label{alg:main}

%\SetAlgoLined
\LinesNumbered
\everypar={\nl}

% \hrule height\algoheightrule\kern3pt\relax
\KwIn{Optimal solution to \eqref{eq: primal}, expressed as in \eqref{eq: rewrite_LP_solution}: $S_r(t),y_r(t)$ for $t\in[T],r\in[n(t)]$}
\KwOut{Schedule of orders $O_t$ for $t\in[T]$}

\vspace{3mm}

Preprocessing: Update $(S_r(t),y_r(t))_{t\in[T],r\in[n(t)]}$ so that it satisfies both \cref{property: nested} and \ref{property:LP_solution_reduced}, as detailed in \cref{appendix:preprocessing}

Continuous time extension: for $s\in(t-1,t]$, $n(s):=n(t)$, and for $r\in[n(s)]$, $S_r(s):=S_r(t)$, $y_r(s):=y_r(t)$

\For{$j\in [k]$}{
    Construct the schedule of orders $\Ocal^{(j)}$ for category $j$ using Algorithm~\ref{alg:water_filling}
}

For $t\in[T]$, let $O_t := \bigcup\set{S : \exists j\in[k],\exists t'\in(t-1,t]: (S,t')\in\Ocal^{(j)}}$

\Return $(O_t)_{t\in[T]}$

% \hrule height\algoheightrule\kern3pt\relax
\end{algorithm}

\subsection{The final schedule of orders.}\label{subsection: final schedule}
Having constructed the order schedules $\Ocal^{(j)}$ for each category $j\in[k]$, we now combine them into a single schedule. At each time $t\in[T]$, the final order $O_t$ is defined as the union of all subsets ordered during the interval $(t-1,t]$ across all categories. Formally, 
\begin{equation*}
    O_t := \bigcup\set{S : \exists j\in[k] ,\exists t'\in(t-1,t], (S,t')\in\Ocal^{(j)}}.
\end{equation*}
The complete algorithm is summarized in Algorithm~\ref{alg:main}, and it runs in polynomial time, as shown in \cref{lemma: runtime}.

\section{Correctness of the DICP Approximation Algorithm}\label{section: correctness}

This section establishes that Algorithm~\ref{alg:main} returns a feasible integral solution to the DICP and bounds its cost. \cref{subsection: feasibility of alg for decomposition} shows that the schedule of orders produced by the algorithm satisfies all demand points. 
\cref{subsection: bounding cost} proves that the ordering cost associated with each of the $k$ categories is at most $4C^\star$ where $C^\star$ is the optimal objective value of the LP relaxation in \eqref{eq: primal}. By submodularity, this yields a bound of $4k C^\star$ on the total ordering cost.

\subsection{Feasibility}\label{subsection: feasibility of alg for decomposition}

The goal of this subsection is to show that the schedule of orders $(O_t)_{t\in[T]}$ produced by Algorithm~\ref{alg:main} is a feasible integral solution to the DICP by establishing the following result.
\begin{proposition}\label{prop:feasibility}
    The schedule of orders $(O_t)_{t\in[T]}$ returned by Algorithm~\ref{alg:main} is feasible for the DICP.
\end{proposition}

Recall that the algorithm constructs a collection of level-sets \( R \in \Lcal \), each with area exactly $1/4$, and guarantees that every such level-set is satisfied. That is, for each level-set $R\in\Lcal$, one of the subsets within that level-set is ordered. The proof of feasibility proceeds in two steps. 

First, it is shown that every demand point $(i,c,d)\in\Dcal$ is associated with a level-set of area at least 1 spanning its demand interval, such that every subset within this level-set contains item $i$. In other words, for every demand point $(i,c,d)$, there exists a threshold index such that every subset ordered by the LP solution during the demand interval has an index below this threshold if and only if it contains item $i$. 
It is then shown that for each demand point, the subsets ordered by the LP solution to satisfy this demand point fully contains at least one of the constructed level-sets \( R \in \Lcal \) (\cref{lemma: contains partition}). Since every constructed level-set in $\Lcal$ is satisfied, the demand point is also satisfied. Together, these steps establish that all demand points are fulfilled by the orders produced by the algorithm. The first step of the proof is captured by the following lemma.

\begin{lemma}\label{lemma: consistent levels}
    Fix any demand point $(i,c,d)\in \Dcal$ and suppose that $i$ belongs to category $j\in[k]$. Let $I:=(c-1,d]$ denote its demand interval. Then, the collection of subsets ordered by the LP solution that contain item $i$ during the interval $I$ forms a level-set for category $j$. Precisely, there exists an index $\alpha \in\mathbb N$ such that every subset in $R := \Rect_j(\alpha,1; I)$ contains item $i$. Moreover, $\Area_j(R) \geq 1$.
\end{lemma}
\begin{proof}[Proof of \cref{lemma: consistent levels}]
    Fix a demand point $(i,c,d)\in\mathcal D$, and let $I:=(c-1,d]$ denote its demand interval.
    Let $\mathbb{S}$ be the collection of pairs $(S,t)$, where $S$ is ordered by the LP solution at time $t$ within the demand interval $I$ and contains item $i$ (i.e., $i\in S$ and $t\in I$). Define
    \begin{equation}\label{eq:def_alpha_i_t}
        \bar{\alpha}:=\max_{(S,t)\in\mathbb{S}} \alpha_j(S,t)
    \end{equation}
    to be the maximum index among these subsets. By construction, we immediately have $\mathbb{S}\subseteq \Rect_j(\bar{\alpha},1;I)$. It remains to show the reverse inclusion. To show $\Rect_j(\bar{\alpha},1;I)\subseteq \mathbb{S}$, it suffices to show that for any $t\in I$ and any subset $S=S_r(t)$ with index $\alpha_j(S,t)\leq \bar{\alpha}$, we must have $i\in S$. By the monotonicity property of indices from \cref{lemma:monotonicity_index}, $A_{t}^i$ is the subset ordered at period $t$ which does not contain $i$ with the smallest index. Hence, it suffices to show that
    \begin{equation}\label{eq:to_be_shown}
        \min_{t\in I} \alpha_j(A_{t}^i,t) >\bar{\alpha}.
    \end{equation}
    
    Let $t_1\in I$ be the time that achieves the minimum in \cref{eq:to_be_shown}, and let $t_2\in I$ be the time that achieves the maximum in \cref{eq:def_alpha_i_t}. Note that $\sum_{S:i\in S} y_{t_2}^S>0$. 
    Next,
    \begin{align}
L_j(A_{t_1}^i)&\overset{(i)}\geq \frac{K(A_{t_1}^i\cup\{i\}) - K(A_{t_1}^i)}{w_i} \notag \\
        &\overset{(ii)}{\geq} \frac{K(B_{t_2}^i) - K(B_{t_2}^i\setminus \{i\})}{w_i} \notag\\
        % &\overset{(iii)}{\geq } \left.\frac{\partial f}{\partial u_j} \right|_{\mathbf u(B_{t_2}^i)}\\
        &\overset{(iii)}{\geq} L_j(B_{t_2}^i). \label{eq:computations_concavity}
    \end{align}
    In $(i)$, the component-wise concavity of \( f \) is used. In $(ii)$ we used \cref{claim:optimality_condition_dual_variable}. In $(iii)$, the component-wise concavity of \( f \) is used again.

    Now suppose for contradiction that \cref{eq:to_be_shown} does not hold, that is, $\alpha_j(A_{t_1}^i,t_1) < \alpha_j(B_{t_2}^i,t_2)$ (equality cannot occur since all subsets have distinct indices). Then, by the definition of the index, $L_j(A_{t_1}^i) \leq L_j(B_{t_2}^i)$, which implies that all inequalities in \cref{eq:computations_concavity} hold with equality. In particular
    \begin{equation}\label{eq:contradiction_eq}
        K(A_{t_1}^i\cup\{i\}) - K(A_{t_1}^i) = K(B_{t_2}^i) - K(B_{t_2}^i\setminus \{i\}).
    \end{equation}
    
If $t_1> t_2$, then combined with $L_j(A_{t_1}^i) = L_j(B_{t_2}^i)$, we would have $\alpha_j(A_{t_1}^i,t_1) > \alpha_j(B_{t_2}^i,t_2)$, a contradiction. If instead $t_1= t_2$, then we would still have $\alpha_j(A_{t_1}^i,t_1) > \alpha_j(B_{t_2}^i,t_2)$ due to the fact that $A_{t_1}^i$ and $B_{t_2}^i$ have equal levels, but $B_{t_2}^i$ is the larger subset out of the two and thus precedes $A_{t_1}^i$ in the lexicographical order, and we again arrive at a contradiction.
Thus, we must have $t_1<t_2$.

However, $t_1<t_2$, combined with \cref{eq:contradiction_eq}, contradicts the fact that the solution $(x,y)$ satisfies Property~\ref{property:LP_solution_reduced} since $\sum_{S:i\in S} y_{t_2}^S>0$.
    Therefore, \cref{eq:to_be_shown} holds, which establishes that $\mathbb{S} = \Rect_j(\bar{\alpha},1;I)$. 
    
    Finally, the area of this level-set corresponds to the total order quantity within $\mathbb{S}$, which is guaranteed to be at least 1 due to the constraint in \eqref{eq: primal}:
    \begin{equation*}
        \Area_j(\Rect_j(\bar{\alpha},1;I)) = \int_{s\in I} \paren{\sum_{S:i\in S} y_s^S} ds \geq 1.
    \end{equation*}
    This completes the proof.
\Halmos\end{proof}

\cref{lemma: consistent levels} shows that for any demand point $(i,c,d)\in \Dcal$, the ordered sets that include $i$ within the demand interval correspond exactly to the ordered subsets within this interval with index up to some threshold $\alpha$. This implies that for any period in this demand interval, the algorithm first fills all ordered sets that include $i$ before filling any sets that do not include $i$. With this observation, we are ready to show the second step of the proof. 

\begin{lemma}\label{lemma: contains partition}
Fix any demand point $(i,c,d)\in\Dcal$ and suppose that $i$ belongs to category $j\in[k]$. The collection $R$ of subsets ordered by the LP solution to satisfy this demand point fully contains at least one constructed level-set from some layer for category $j$. That is, there exists $R'\in\Lcal(j)$ such that $R'\subseteq R$. 
\end{lemma}
\begin{proof}[Proof of \cref{lemma: contains partition}]
By feasibility, $\Area_j(R)\geq 1$. In other words, the total area of ordered subsets containing item $i$ within $(c-1,d]$ is at least 1, so the water-filling partitioning algorithm is guaranteed to make a split within this demand interval.

Consider the first time the algorithm makes a split within the demand interval $(c-1,d]$. Suppose that the split is triggered when a total area of $1/2$ has been filled over some interval $(a,b]$ that fully contains the demand interval, and suppose the split occurs at period $s$. This partitions the demand interval $(c-1,d]$ into two sub-intervals: $(c-1,s]$ and $(s,d]$. Since the total area of ordered subsets containing item $i$ within $(c-1,d]$ is at least 1, at least one of the two resulting sub-intervals ($(c-1,s]$ and $(s,d]$) must contain area at least $1/2$ from ordered subsets that include item $i$. Without loss of generality, suppose this is the right sub-interval $(s,d]$.

Because $(s,d]$ contains at least $1/2$ area and has not yet been split, the algorithm must eventually make another split within this interval. This split is triggered when a total area of $1/2$ has been filled over some interval $(s,s']$, where $b\geq s' \geq d$. Therefore, the portion of this filled area that lies within $(s,d]$ is at most $1/2$.

Moreover, within $(s,d]$, subsets containing item $i$ are always filled before any subsets that do not contain $i$. Since $(s,d]$ already contains at least $1/2$ area from subsets containing $i$, then all filled subsets within $(s,d]$ must contain $i$. It follows that one of the level-sets created by this split lies entirely within $(c-1,d]$ and consists only of ordered subsets that include item $i$, as claimed.
\Halmos
\end{proof}

By the design of Algorithm~\ref{alg:main}, every level-set \( R \in \Lcal \) is satisfied. Hence, the two lemmas above directly imply that the ordering schedule $(O_t)_{t \in [T]}$ is a feasible solution. This concludes the proof of \cref{prop:feasibility}.

\subsection{Bounding the cost of the orders}\label{subsection: bounding cost}

To bound the cost of the ordering schedule \( \Ocal^{(j)} \) produced by the dynamic program for each category $j\in [k]$, an alternate randomized algorithm is analyzed. Specifically, instead of the dynamic program in Step 2 of Algorithm~\ref{alg:water_filling}, the randomized procedure samples ordering pairs \( (S, t) \) to satisfy all the level-sets \( R \in \Lcal \). It is shown that the expected total ordering cost of this sampled schedule is at most four times the LP cost in \eqref{eq: primal}. Since the dynamic program computes the minimum-cost solution that satisfies all level-sets, its cost cannot exceed the expected cost of the randomized procedure. Therefore, the cost of the dynamic program solution is also at most four times the cost of the LP. This section provides an intuitive description of the randomized procedure and the proof for bounding its expected cost; the formal algorithm and proof are given in  \cref{subsection: appendix_expected_cost}.

Recall that in the water-filling procedure of Algorithm~\ref{alg:main}, each category \( j \in [k] \) is handled independently. The randomized procedure for constructing an ordering schedule is now described for a fixed category \( j \in[k]\). 

The procedure proceeds layer by layer, starting from the first layer $\Lcal_1$. For each level-set in $\Lcal_1$, a subset is independently sampled with probability proportional to its area within the level-set, where partially included subsets are sampled proportionally to their extent of inclusion. If a subset $S$ ordered at time $t$ is selected, the order $(S,t)$ is added to the schedule and all descendant level-sets whose intervals contain $t$ are marked as satisfied; these level-sets are then removed from further consideration.

The procedure then continues to subsequent layers. For each level-set that has not yet been satisfied, the same sampling procedure is repeated. Importantly, when sampling from a level-set $\Rect_j(\alpha,\beta;I)\in \Lcal_\ell$, the sampling distribution is restricted to the portion of the level-set that is not already covered by its parent $\Rect_j(\alpha_p,\beta_p;I_p)\in \Lcal_{\ell-1}$. That is, sampling occurs only from the region of $\Rect_j(\alpha,\beta;I)$ lying outside its parent. Again, the selected order $(S,t)$ is added to the schedule, and all descendant level-sets containing $t$ are removed.

By construction, this procedure ensures that every level-set is satisfied by at least one order. The randomized algorithm is formalized in Algorithm~\ref{alg:randomized} in \cref{subsection: appendix_expected_cost}, and it achieves cost at most four times that of the optimal objective value $C^\star$ of the LP relaxation \cref{eq: primal}.

\begin{proposition}\label{prop:cost_each_order}
   Fix a category $j\in[k]$, and let $\tilde\Ocal^{(j)}$ denote the schedule of orders returned by the randomized procedure (which is formally described in Algorithm~\ref{alg:randomized}) for category $j$. Then,
   \begin{equation*}
       \Ebb\sqb{\sum_{(S,t)\in\tilde\Ocal^{(j)}} K(S)} \leq 4C^\star.
   \end{equation*}
\end{proposition}

Instead of presenting the full formal proof (deferred to \cref{subsection: appendix_expected_cost}), we provide intuition for the result. Step 1 of Algorithm~\ref{alg:water_filling} constructed a nested collection of level-sets with \emph{constant LP mass per region}; specifically, each region carries LP mass $1/4$. However, within each region, the algorithm selects a subset to order with probability 1. This mismatch—ordering with probability 1 from regions with LP mass $1/4$—is what contributes the factor of 4 in the cost bound of \cref{prop:cost_each_order}.

To relate the algorithm’s behavior to the LP, we aim to show that the probability with which any given subset is ordered is proportional, up to a factor of 4, to its LP ordering quantity. Equivalently, we show that the algorithm samples from any region (i.e., the portion of a level-set not covered by its parent) with probability equal to 4 times its area. Establishing this requires carefully accounting for overlap between level-sets and the fact that some level-sets may already be satisfied by orders selected in earlier layers.

We verify this claim by induction on the layers. In the base case, the algorithm samples from every level-set in the first layer. Each such region has area $1/4$ and we sample from each with probability 1, so the sampling probability is exactly $4 \times \Area(\cdot)$, as desired.

For the inductive step, assume the claim holds for all layers up to $\ell-1$, and consider a level-set $R$ in layer $\ell$. Let $Q$ denote the portion of its parent that overlaps with $R$. The algorithm samples from $R \setminus Q$ only if no subset from $Q$ has been ordered previously. By the induction hypothesis, the probability that an order has already been placed in $Q$ is $4 \cdot \Area(Q)$. Therefore, the probability that a subset is sampled from $R \setminus Q$ is $1-4\cdot \Area(Q),$ which is precisely $4 \cdot \Area(R \setminus Q)$ since $\Area(R) = 1/4$. This establishes the desired proportionality and completes the inductive argument.

Since Algorithm~\ref{alg:water_filling} computes a minimum-cost schedule that satisfies all level-sets, we can use \cref{prop:cost_each_order} to bound the cost of the resulting orders for each category $j\in[k]$. Then, using submodularity, we obtain the following overall cost bound:

\begin{proposition}\label{prop:final_cost}
    The cost for the orders $(O_t)_{t\in[T]}$ returned by Algorithm~\ref{alg:main} is at most $4k C^\star$.
\end{proposition}

\begin{proof}[Proof of \cref{prop:final_cost}]
    Fix a category $j\in[k]$. In step 2 of Algorithm~\ref{alg:water_filling}, a dynamic program computes the minimum-cost schedule $\Ocal^{(j)}$ that satisfies all level-sets across all layers $\bigcup_{l\geq 1}\Lcal_l$ constructed in step 1 of Algorithm~\ref{alg:water_filling}. By design, this cost is at most the expected cost of the randomized schedule of orders constructed using Algorithm~\ref{alg:randomized}, denoted $\tilde \Ocal^{(j)}$. Hence, \cref{prop:cost_each_order} implies that 
    \begin{equation}\label{eq:cost_category_j_final}
        \sum_{(S,t)\in\Ocal^{(j)}} K(S) \leq \Ebb \sqb{\sum_{(S,t)\in\tilde \Ocal^{(j)}} K(S)} \leq 4C^\star.
    \end{equation}
    
    Let $O_t=\bigcup\{S:\exists j\in[k],\exists t'\in(t-1,t]:(S,t')\in\Ocal^{(j)}\}$ denote the final order placed at time $t\in[T]$ by Algorithm~\ref{alg:main}.
    By the submodularity of the cost function $K$, we have
    \begin{equation*}
        K(O_t)\leq \sum_{j\in[k]} \sum_{\substack{(S,t')\in \Ocal^{(j)}:\\ t'\in(t-1,t]}} K(S).
    \end{equation*}
    Summing over all time periods $t\in[T]$, 
    \begin{equation*}
        \sum_{t\in[T]}K(O_t) \leq \sum_{j\in[k]} \sum_{(S,t)\in\Ocal^{(j)}}K(S) \leq 4kC^\star,
    \end{equation*}
where the last inequality follows from \eqref{eq:cost_category_j_final}.
This completes the proof.
\Halmos\end{proof}

%The state space consists of all level-sets \( R \in \Lcal \), which we have previously shown to be polynomial in number. Each DP step involves minimizing over subsets contained in \( R \), and there are only polynomially many such subsets due to the polynomial-size LP solution. The summation over all descendants of \( R \) can also be computed in polynomial time, as each level-set has a polynomial number of descendants given the overall polynomial number of regions.

\section{SJRP Approximation Algorithm} \label{section: Reduction}
This section describes how to solve the original SJRP by leveraging the approximation algorithm for the simpler DICP, presented in \cref{section: algorithm}, as a subroutine.
The main result shows that the SJRP with cost functions of the form \cref{eq:form_cost} admits an efficient approximation algorithm, with an approximation guarantee that scales with the number of item categories $k$.

\begin{theorem}\label{thm:main}
    Suppose that the ordering cost function $K$ admits a decomposition of the form \cref{eq:form_cost}. Then Algorithm~\ref{alg:Reductions} produces a solution to the SJRP in polynomial time with an $(8k+1)$-approximation guarantee. 
\end{theorem}

We begin with an optimal solution to (LP), the LP relaxation of the SJRP, whose objective value provides a lower bound on the optimal integral solution. Using this solution, a corresponding DICP instance is constructed such that the \textit{ordering cost} incurred by the fractional solution to its LP relaxation is at most twice that of (LP). In \cref{section: algorithm}, it was shown how to round this fractional DICP solution to an integer one, with an ordering cost at most $4k$ times larger, where $k$ is the number of item categories from the decomposition given in \cref{def:form_cost}. 
This section shows that the integral solution obtained for the DICP yields a feasible solution to the original SJRP instance and bounds both its ordering and holding cost. %Moreover, it establishes that the holding cost incurred by this integral solution is bounded by the cost of the LP relaxation of the SJRP.

% with ordering cost at most $8k$ times the cost of (LP), and a holding cost bounded by it.

\subsection{Algorithm}\label{subsection: Reduction alg}
This subsection describes how to construct a DICP instance $\mathcal{I}_{\text{DICP}}$ from a given SJRP instance $\mathcal{I}$. Intuitively, for each item, a sequence of time intervals is identified based on cumulative fractional order quantities reaching 0.5. These intervals become the demand intervals for the corresponding item in the constructed DICP instance. This intuition is formalized below.

Fix an optimal LP solution $(x,y)$ to the linear relaxation (LP) of $\Ical$. For convenience, we use a continuous relaxation of the time periods $[T]$: each discrete time period $t$ is viewed as the continuous interval $(t-1,t]$. In particular, the entire time horizon becomes $(0,T]$. We then naturally extend the LP solution to this continuous setting by defining $x_{rd}^i=x_{sd}^i$ and $y_{r}^S=y_s^S$ for all $r\in(s-1,s]$, for each demand point $(i,d)$ and item subset $S\subseteq [N]$. 

Fix an item $i\in[N]$ and define $z_s^i := \sum_{S:i\in S} y_s^S$ for $s\in(0,T]$, which corresponds to the total fractional order of item $i$ in period $s$. We can also view $z_s^i$ as the density of a measure over the interval $(0,T]$. Next, we define the set of all $1/2$-quantiles of this measure. Precisely, let $K^i=\lfloor 2\int_0^T z_s^i ds \rfloor$ denote the number of such quantiles within the time horizon. For each $\kappa \in\{0,\ldots,K^i\}$, we define the lower and upper $\kappa /2$-quantile as follows:
\begin{align}\label{eq:definition_a_b}
    &a_\kappa^i:= \max Q_\kappa^i, \quad b_\kappa^i:= \min Q_\kappa^i, \\
    &\text{where}\quad Q_\kappa^i:=\set{t\in[0,T]: \int_{s=0}^{t} z_s^i ds = \frac{\kappa}{2}}.\notag
\end{align}

% We now define the reduced DICP instance, denoted $\Ical_{DICP}$. For each item $i\in[N]$, the demand intervals are given by the intervals $(a_{\kappa-1}^i,b_\kappa^i]$ for all $\kappa\in[K^i]$. We have now obtained an instance for which the demand intervals are in continuous time. To turn it into a discrete problem, recall that each discrete time period $t$ in the original instance $\Ical$ from the SJRP was associated to the continuous interval $(t-1,t]$. We therefore map each continuous-time interval $(a_{\kappa-1}^i,b_\kappa^i]$ to the discrete-time interval $\sqb{\lfloor a_{\kappa-1}^i\rfloor+1,\lceil b_\kappa^i \rceil}$ for all $i\in[N]$ and $\kappa\in[K^i]$. This concludes the construction of the reduced instance $\Ical_{DICP}$.

We now define the reduced DICP instance, denoted $\Ical_{DICP}$. The reduced instance is an ordinary discrete-time DICP instance, but its periods are obtained from a refinement of the original time horizon. Define
\begin{equation}\label{eq:refined_grid}
\Ecal:=\{0,1,\ldots,T\}\cup\{a_\kappa^i,b_\kappa^i:i\in[N],\kappa\in[K^i]\}.
\end{equation}
Let $0=\tau_0<\tau_1<\cdots<\tau_M=T$ denote the sorted elements of $\Ecal$, and let $J_m:=(\tau_{m-1},\tau_m]$ for $m\in[M]$. Since $\{0,1,\ldots,T\}\subseteq\Ecal$, each atom $J_m$ is contained in a unique original period $(t(m)-1,t(m)]$, where $t(m)\in[T]$.

For each item $i\in[N]$ and $\kappa\in[K^i]$, let $c_\kappa^i,d_\kappa^i\in[M]$ be such that
\begin{equation}\label{eq:atom_index_interval}
(a_{\kappa-1}^i,b_\kappa^i]=\bigcup_{m=c_\kappa^i}^{d_\kappa^i}J_m.
\end{equation}
The reduced DICP instance has horizon $[M]$ and contains the demand interval $(i,c_\kappa^i,d_\kappa^i)$ for every $i\in[N]$ and $\kappa\in[K^i]$. Since the continuous intervals $(a_{\kappa-1}^i,b_\kappa^i]$ are disjoint for each fixed item $i$, the corresponding discrete intervals $[c_\kappa^i,d_\kappa^i]$ are also disjoint. This concludes the construction of the reduced instance $\Ical_{DICP}$.

\begin{algorithm}[t]

\caption{Solving an SJRP Instance via Reduction to a DICP Instance}\label{alg:Reductions}

%\SetAlgoLined
\LinesNumbered
\everypar={\nl}

% \hrule height\algoheightrule\kern3pt\relax
\KwIn{SJRP instance}
\KwOut{Feasible ordering schedule for the SJRP instance}

\vspace{3mm}

Compute an LP solution $x^i_{sd}, y_s^S$ for $(i,d)\in \mathcal D, s\in [T], S\subseteq [N]$ to the SJRP instance

Continuous time extension: 
$x_{rd}^i:=x_{sd}^i$ and $y_{r}^S:=y_s^S$ for $(i,d)\in \mathcal D, s\in [T], r\in(s-1,s], S\subseteq [N]$

For each $i\in[N]$, let $z_s^i:=\sum_{S:i\in S} y_s^S$ for $s\in(0,T]$, and let $K^i:=\lfloor 2\int_0^T z_s^i ds \rfloor$

For each $i\in[N]$ and $\kappa\in[K^i]$, define $a_\kappa^i,b_\kappa^i$ as in \cref{eq:definition_a_b}

Construct the reduced DICP instance $\Ical_{DICP}$ using \cref{eq:refined_grid,eq:atom_index_interval}

Run Algorithm~\ref{alg:main} on $\Ical_{DICP}$ and obtain the output schedule $(\tilde O_m)_{m\in[M]}$

For each $t\in[T]$, define $\widehat O_t:=\bigcup_{m:t(m)=t}\tilde O_m$

Define $O_s:= \set{i\in[N]: \exists (i,d)\in\Dcal \text{ such that } 
s=\max\{r\leq d: i\in \widehat O_r\}}$

\Return $(O_s)_{s\in[T]}$

% \hrule height\algoheightrule\kern3pt\relax
\end{algorithm}

It remains to describe how to reconstruct a solution to $\Ical$ from a solution to $\Ical_{DICP}$. Let $(\tilde O_m)_{m\in[M]}$ denote a feasible integral schedule of orders for $\Ical_{DICP}$. We first map this atom-indexed schedule back to the original time periods by defining, for every $t\in[T]$,
\[
\widehat O_t:=\bigcup_{m:t(m)=t}\tilde O_m.
\]
To reconstruct a feasible solution for $\Ical$, we assign each demand point $(i,d)$ to the latest order in $(\widehat O_t)_{t\in[T]}$ that includes item $i$ and occurs no later than time $d$. Formally, we define
\[
s_{(i,d)}:=\max\{r\leq d:i\in \widehat O_r\},
\]
and fulfill demand $(i,d)$ at time $s_{(i,d)}$. The resulting schedule of discrete-time orders for $\Ical$ is then defined as follows for every period $s\in[T]$:
\begin{equation}\label{eq:construction_final_schedules}
    O_s:=\set{i\in[N]:\exists (i,d)\in\Dcal,\ s_{(i,d)}=s}.
\end{equation}
The entire procedure is summarized in Algorithm~\ref{alg:Reductions}.

\subsection{Correctness}

We now show that the resulting schedule of orders $(O_s)_{s\in[T]}$, defined in \cref{eq:construction_final_schedules}, is feasible for the SJRP and incurs only a constant-factor loss in the approximation guarantee relative to the cost of the associated DICP instance. This is formalized in the following result:

\begin{proposition}\label{proposition: Reduction}
    Consider any SJRP instance $\Ical$. Algorithm~\ref{alg:Reductions} constructs, in polynomial time, a corresponding DICP instance $\Ical_{DICP}$  such that if the solution for $\Ical_{DICP}$ has cost at most $\gamma$ times the cost of its LP relaxation, then the schedule constructed in \cref{eq:construction_final_schedules} is a $(2\gamma+1)$-approximation for $\mathcal{I}$.
\end{proposition}

Since the LP relaxation of an SJRP instance can be solved in polynomial time, it is readily verified that Algorithm~\ref{alg:Reductions} also constructs the DICP instance in polynomial time. Indeed, the refined grid contains only the original integer endpoints and the quantile endpoints generated by the reduction, so $M$ is polynomial in the input size and the support size of the LP solution. To complete the proof of \cref{proposition: Reduction}, we proceed in three steps: we establish feasibility in \cref{lemma: Reduction_feasible}, bound the holding cost in \cref{lemma: Reduction_holding}, and bound the ordering cost in \cref{lemma: Reduction_ordering} and \cref{lemma: atom_mapping_cost}. The proofs of all these lemmas are provided in Appendix~\ref{section: appendix Reduction}.

% To show that the schedule constructed in \cref{eq:construction_final_schedules} yields a feasible solution to the original SJRP instance $\Ical$, we show that for each demand point $(i,d)$ in $\Ical$, there is a demand interval in $\Ical_{DICP}$ that lies entirely within its \textit{active interval}---the smallest interval containing all times during which the LP solution places fractional orders to fulfill this demand $\set{s: x^i_{sd}>0}$. Since $i$ is guaranteed to be ordered within each of its demand intervals, any feasible solution to $\Ical_{DICP}$ is also feasible for $\Ical$.

To show that the schedule constructed in \cref{eq:construction_final_schedules} yields a feasible solution to the original SJRP instance $\Ical$, we show that for each demand point $(i,d)$ in $\Ical$, there is a demand interval in $\Ical_{DICP}$ whose associated atoms lie within the continuous interval $(s^\star_{(i,d)},d]$. Since the DICP solution orders item $i$ in one of these atoms, the mapped schedule $(\widehat O_t)_{t\in[T]}$ orders item $i$ no later than period $d$.

\begin{lemma}\label{lemma: Reduction_feasible}
For any demand point $(i,d)$ in $\Ical$, there exists $\kappa\in[K^i]$ such that
\begin{equation}\label{eq:important_property_Reduction}
        (a_{\kappa-1}^i,b_\kappa^i]\subseteq (s^\star_{(i,d)},d],
\end{equation}
where $s^\star_{(i,d)}:=\inf\set{s\in(0,d]:x_{sd}^i>0}$. Hence, the schedule $(O_s)_{s\in[T]}$ satisfies all demand points in $\Ical$.
\end{lemma}

We next bound the holding cost. \cref{lemma: Reduction_holding} shows that the holding cost incurred by the schedule $(O_s)_{s\in[T]}$ is at most the total cost of the optimal LP solution $(x,y)$ to (LP), the LP relaxation of $\Ical$. This result also appears in \cite{cheung2016submodular}; for completeness, we include a proof in the appendix based on a weak duality argument.

\begin{lemma}\label{lemma: Reduction_holding}
The holding cost incurred by the schedule $(O_s)_{s\in[T]}$ is at most the total cost of the optimal LP solution $(x,y)$ to (LP), the LP relaxation of $\Ical$.
\end{lemma}

We next bound the ordering cost in the LP relaxation of $\mathcal{I}_{DICP}$ in \cref{lemma: Reduction_ordering}, which follows by observing that scaling $(x, y)$ by a factor of two yields a feasible solution to this LP relaxation.

\begin{lemma}\label{lemma: Reduction_ordering}
The ordering cost in the LP relaxation of $\Ical_{DICP}$ is at most twice that of the optimal LP solution $(x,y)$ to (LP), the LP relaxation of $\Ical$.
\end{lemma}

The previous lemma bounds the LP cost of the reduced DICP instance relative to the LP cost of the original SJRP instance. It remains to check that converting an atom-indexed DICP schedule back to the original time horizon does not introduce any additional ordering cost. This is because several atom-level orders may be mapped to the same original period, but they are consolidated into a single order in that period.

\begin{lemma}\label{lemma: atom_mapping_cost}
Let $(\tilde O_m)_{m\in[M]}$ be any feasible integral schedule for $\Ical_{DICP}$, and let $(O_t)_{t\in[T]}$ be the corresponding SJRP schedule constructed in Algorithm~\ref{alg:Reductions}. Then
\[
    \sum_{t=1}^T K(O_t)
    \leq
    \sum_{m=1}^M K(\tilde O_m).
\]
\end{lemma}

Together, these lemmas establish the approximation guarantee claimed in \cref{proposition: Reduction}. By \cref{lemma: Reduction_ordering}, the LP relaxation of the reduced DICP instance has ordering cost at most twice the ordering cost of the optimal LP solution to $\Ical$. Therefore, if the DICP algorithm achieves an ordering-cost approximation ratio of $\gamma$, then
\[
\sum_{m=1}^M K(\tilde O_m)
\leq 2\gamma \sum_{s=1}^T\sum_{S\subseteq[N]}K(S)y_s^S.
\]
By \cref{lemma: atom_mapping_cost}, mapping the atom-indexed DICP schedule back to the original SJRP periods does not increase ordering cost, and hence
\[
\sum_{t=1}^T K(O_t)
\leq 2\gamma \sum_{s=1}^T\sum_{S\subseteq[N]}K(S)y_s^S.
\]
Moreover, by \cref{lemma: Reduction_holding}, the holding cost of the constructed schedule is at most the total objective value of the optimal LP solution to $\Ical$. It follows that Algorithm~\ref{alg:Reductions} gives a $(2\gamma+1)$-approximation for the SJRP instance $\Ical$.

Finally, all results are combined to establish \cref{thm:main}. \Cref{prop:feasibility} guarantees that the orders returned by Algorithm~\ref{alg:main} form a feasible solution to the DICP, while \cref{prop:final_cost} bounds the total ordering cost across all $k$ categories by at most $4k$ times the LP relaxation cost of the DICP. Substituting $\gamma=4k$ into \cref{proposition: Reduction} then yields an overall $(8k+1)$-approximation for the SJRP, as claimed in \cref{thm:main}.

\section{Conclusions}\label{section: conclusion}

This paper presents a novel approximation algorithm for the submodular joint replenishment problem (SJRP) under a broad family of submodular ordering cost functions that admit a structured decomposition over item categories. This framework substantially generalizes prior work by capturing complex cross-category interactions through an arbitrary multivariate function over category-level aggregates. Our main result establishes an $O(k)$-approximation guarantee, where $k$ is the number of item categories. For fixed $k$, this yields the first known constant-factor bound for this general class—a guarantee that was previously unknown even in the case $k=1$.

The analysis proceeds by recasting the problem as a disjoint interval covering problem (DICP), where the objective is to select orders that cover disjoint demand intervals. A novel rounding algorithm is introduced, combining a water-filling partitioning of the fractional LP solution with a dynamic programming procedure to produce a feasible integral schedule. This approach departs substantially from prior methods.

We believe the techniques developed here may inform future work on approximation algorithms for inventory problems, including the longstanding open problems of obtaining constant-factor guarantees for the SJRP under general submodular cost functions and the IRP under arbitrary metrics.

% Acknowledgments here
\section*{Acknowledgments}
This work was supported by the National Science Foundation Graduate Research Fellowship under Grant No. 2141064. Any opinions, findings, conclusions, or recommendations expressed in this material are those of the authors and do not necessarily reflect the views of the National Science Foundation.

% References here (outcomment the appropriate case)

% CASE 1: BiBTeX used to constantly update the references
%   (while the paper is being written).
\bibliographystyle{plainnat} % outcomment this and next line in Case 1
\bibliography{sample} % if more than one, comma separated

%\bibliographystyle{plainnat} % outcomment this and next line in Case 1
%\bibliography{sample} % if more than one, comma separated

% CASE 2: BiBTeX used to generate mypaper.bbl (to be further fine tuned)
%\input{mypaper.bbl} % outcomment this line in Case 2

%If you don't use BiBTex, you can manually itemize references as shown below.

%\bibliographystyle{nonumber}

% \begin{thebibliography}{3}
% \providecommand{\natexlab}[1]{#1}
% \providecommand{\url}[1]{\texttt{#1}}
% \providecommand{\urlprefix}{URL }

% \bibitem[{Smith(2005)}]{smith2005}
% Smith J (2005) Optimal resource allocation in humanitarian logistics.
%   \emph{Journal of Operations Research} 30(2):123--135.
  
% \bibitem[{Jones(2010)}]{jones2010}
% Jones S (2010) Stochastic programming models for humanitarian logistics.
%   \emph{INFORMS Mathematics of Operations Research} 35(4):567--580.

% \bibitem[{Brown(2015)}]{brown2015}
% Brown D (2015) \emph{Introduction to Stochastic Programming} (Springer).

% \end{thebibliography}

% Appendix here
% Options are (1) APPENDIX (with or without general title) or
%             (2) APPENDICES (if it has more than one unrelated sections)
% Outcomment the appropriate case if necessary
%

\newpage
% \begin{APPENDIX}{Missing Proofs}

%
%   or
% %
\appendix
\section{Proof of Correctness for the DICP Approximation Algorithm} \label{section: appendix approx alg}

We include here the proofs omitted from the main text that justify the approximation guarantee for the DICP.

Below, we prove that the notion of level and index introduced in \cref{def:level,def:index} follow the desired monotonicity property.

\begin{proof}[Proof of \cref{lemma:monotonicity_levels}]
    It suffices to show that for any subset of items $S$ and $i\notin S$, we have $L_j(S\cup\{i\}) \leq L_j(S)$. If $P_j\subseteq S$ the result is immediate from $L_j(S\cup\{i\}) = L_j(S)=0$. We now suppose that this is not the case and distinguish between two cases.

    First, suppose that $i\notin \argmin\set{w_{i'},i'\in P_j\setminus S}$. Fix $i_0\in \argmin\set{w_{i'},i'\in P_j\setminus S}$. Deleting $i$ does not change this argmin, hence we also have $i_0\in \argmin\set{w_{i'},i'\in P_j\setminus (S\cup\{i\})}$ Then,
    \begin{align*}
        &L_j(S\cup\{i\}) \\
        &= \frac{f(\mathbf u(S\cup\{i\}) + w_{i_0} \cdot \mathbf e_j) - f(\mathbf u(S\cup\{i\})) }{w_{i_0}} \\
        &= \frac{K(S\cup\{i_0,i\})- K(S\cup\{i\})) }{w_{i_0}} \\
        &\leq \frac{K(S\cup\{i_0\})- K(S)) }{w_{i_0}} \\
        &= L_j(S).
    \end{align*}
    In the inequality, we used the fact that $K$ is submodular.

    Next, consider the case where $i\in  \argmin\set{w_{i'},i'\in P_j\setminus S}$. Then, 
    \begin{align*}
        &L_j(S) \\
        &= \frac{f(\mathbf u(S) + w_i \cdot \mathbf e_j) - f(\mathbf u(S)) }{w_i}\\
        &\overset{(i)}{\geq} \left.\frac{\partial f}{\partial u_j}\right|_{\mathbf u(S) + w_i \cdot \mathbf e_j}  = \left.\frac{\partial f}{\partial u_j}\right|_{\mathbf u(S\cup\{i\})}\\
        &\overset{(ii)}{\geq} \sup_{w>0}\frac{f(\mathbf u(S\cup\{i\}) + w \cdot \mathbf e_j) - f(\mathbf u(S\cup\{i\})) }{w} \\
        &\geq L_j(S\cup\{i\}).
    \end{align*}
    In $(i)$ and $(ii)$, we used the component-wise concavity of \( f \). 
\Halmos\end{proof}

\begin{proof}[Proof of \cref{claim:optimality_condition_dual_variable}]
    Assume for contradiction that the statement is not true. Then there exist time periods \( t, \tau \in [c,d] \) such that
    \begin{align*}
        K\paren{A^i_t \cup \set{i}} - K\paren{A^i_t}  <  K\paren{B^i_{\tau}} - K\paren{B^i_{\tau}\setminus \set{i}}.
    \end{align*}

    Consider a perturbation of the LP solution where we increase the order quantity of a subset containing item \( i \) at time \( t \) by a sufficiently small amount \( \varepsilon > 0 \), and simultaneously decrease the order quantity of a subset containing \( i \) at time $\tau$ by the same amount \( \varepsilon \), preserving feasibility.

The change in the total cost from this perturbation is:
\[
\varepsilon \left[ K\left(A_t^i \cup \{i\}\right) - K\left(A_t^i\right) - K\left(B_{\tau}^i\right) + K\left(B_{\tau}^i \setminus \{i\}\right) \right].
\]

By the assumption above, the expression in brackets is strictly negative, so the total cost strictly decreases. This contradicts the optimality of the LP solution.
\Halmos\end{proof}

We now prove that Algorithm~\ref{alg:main} has polynomial running time.

\begin{lemma}\label{lemma: runtime}
    Algorithm~\ref{alg:main} has polynomial running time.
\end{lemma}

\begin{proof}[Proof of \cref{lemma: runtime}]
We first verify that the LPs \cref{eq: primal,eq: primal_modified} can be solved in polynomial time. This follows from classical results: their duals are finite-dimensional but have exponentially many constraints, which can be handled with the ellipsoid method. A separation oracle for identifying violated dual constraints reduces to minimizing a submodular function (namely, the cost function $K$ plus a linear term). Since submodular minimization is polynomial-time solvable, and a primal solution can be recovered from the dual in polynomial time \citep{schrijver1998theory}, both LPs are solvable in polynomial time.

%\eqref{eq: primal} can be solved in polymomial time using the ellipsoid method, since a separation oracle for identifying violated dual constraints can be implemented using submodular function minimization. Namely, The same argument applies for solving the preprocessing problem \eqref{eq: primal_modified}. the 

Although Algorithm~\ref{alg:main} is described using a continuous-time extension over $(0,T]$, it can be implemented in a discrete manner. Specifically, it only requires storing information at discrete time points: the original times $t \in [T]$ and the additional split points $t\in(0,T]$ used for partitioning layers. At each such time, it suffices to track the subsets ordered by the LP, the regions at each layer, and their associated levels $\alpha$ and proportion $\beta$ for the last level. Therefore, to show that Algorithm~\ref{alg:main} runs in polynomial time, it suffices to show that Algorithm~\ref{alg:water_filling} constructs a polynomial number of level-set regions (and layers) (this would imply that the dynamic program \cref{eq: DP} would also run in polynomial time). 

To bound the number of level-sets constructed by Algorithm~\ref{alg:water_filling} for a given category $j\in[k]$, consider a parent level-set $R=\Rect_j(\alpha,\beta;I)\in\Lcal_l$ in some layer $l\geq 0$. Its two children collectively cover area exactly $1/2$, while the parent has area $1/4$. Let $P(R):=(R_1\cup R_2)\setminus R$, where $R_1$ and $R_2$ are the children of $R$. For all parent level-sets $R$, these sets $P(R)$ are disjoint and each have area $1/4$. Therefore,
\begin{align}\label{eq:bound_nb_parents}
    &\text{Number of parent level-sets} \\
    &\leq 4\cdot (\text{Total area of LP solution}) \notag\\
    &\leq 4|\Dcal|T.\notag
\end{align}
The final bound follows since for all non-empty $S\subseteq[N]$, we have $K(S)>0$, so if we had $\sum_{S\subseteq[N]}y_t^S>|\Dcal|$ at any time $t\in[T]$, we could strictly reduce the LP cost by decreasing some $y_t^S$ while maintaining feasibility. Since each parent region gives rise to at most two children, the total number of level-sets constructed by Algorithm~\ref{alg:water_filling} is at most $12|\Dcal|T$, which ends the proof.
\Halmos\end{proof}
% \end{APPENDIX}

Finally, we establish conditions on the decomposition function $f$ that guarantee the submodularity of the ordering cost function $K$.

\begin{lemma}\label{lemma: f conditions}
    Let $f:\Rbb^k_+\to\Rbb_+$ be a function that is twice continuously differentiable on $\Rbb^k_+$. If for any $i,j\in [k]$,
        \begin{equation*}
            \forall \mathbf v\in\Rbb_+^k,\quad \frac{\partial^2 f}{\partial x_i\partial x_j}(\mathbf v)\leq 0,
        \end{equation*}
    then, for any weight vectors $\mathbf w_1,\ldots,\mathbf w_k\in\Rbb_+^N$ with strictly positive entries on items in $P_j$ and zeroes elsewhere, the set function $K:S\subseteq[N]\mapsto f(\mathbf w_1^\top \mathbf e_S,\ldots,\mathbf w_k^\top \mathbf e_S)$ is submodular.
    
    Moreover, if $|P_j|\geq 3$ for all categories $j\in[k]$, the two statements are equivalent.
\end{lemma}

We note that $f$ need not be jointly concave for $K$ to be submodular under this structure.

\begin{proof}[Proof of \cref{lemma: f conditions}]
    We start with the first claim. Assume that $\frac{\partial^2 f}{\partial x_i\partial x_j}\leq 0$ for all $i,j\in[k]$, and fix weight vectors $\mathbf w_1,\ldots,\mathbf w_k$ satisfying the lemma's conditions. 
    
    Let $A\subset[N]$ and pick distinct items $i,  b \in[N]\setminus A$. Set $B:=A\cup\{b\}$---in other words, $\mathbf u(B)= \mathbf u(A)+w_b\mathbf e_{j(b)}$, where $j(b)$ denotes the category of item $b$. Then,
    \begin{align}
        &(K(B\cup\{i\})-K(B)) - (K(A\cup\{i\})-K(A))\notag \\
        &= \paren{f(\mathbf u(B)+w_i \mathbf e_{j(i)}) - f(\mathbf u(B))}\notag\\
&\qquad - \paren{f(\mathbf u(A)+w_i \mathbf e_{j(i)}) - f(\mathbf u(A))}  \notag \\
        & =\int_{y=0}^{w_i} \Bigg( \frac{\partial f}{\partial x_{j(i)}}(\mathbf u(B)+ y\mathbf e_{j(i)}) \notag\\
&\qquad\qquad - \frac{\partial f}{\partial x_{j(i)}}(\mathbf u(A)+ y\mathbf e_{j(i)})\Bigg) dy  \notag \\
        &= \int_{y=0}^{w_i} \int_{z=0}^{w_b}  \frac{\partial^2 f}{\partial x_{j(i)} \partial x_{j(b)}}\Big(\mathbf u(A) \notag\\
&\qquad\qquad\qquad \qquad  + y\mathbf e_{j(i)} + z\mathbf e_{j(b)}\Big) dz dy\label{eq:useful_computation_square}\\
        &\leq 0. \notag
    \end{align}
    That is, we showed that
    \begin{equation}\label{eq:submodular_constraint}
        K(B\cup\{i\})-K(B) \leq K(A\cup\{i\})-K(A),
    \end{equation}
    whenever $B$ contains exactly one more item than $A$. Summing these equations then shows that for any $A\subseteq B$ and $i\in[N]$, \cref{eq:submodular_constraint} holds. This shows that $K$ is submodular.

    Next, we prove the second claim assuming that $|P_j|\geq 3$ for all $j\in[k]$. Suppose that the function $f$ satisfies the second statement---that is, for any choice of weights, the corresponding function $K$ is submodular. Fix $\mathbf v\in\Rbb_+^k$, and for each category $j\in[k]$, we fix three distinct items $i_1(j),i_2(j),i_3(j)\in P_j$. Finally, fix two (possibly equal) categories $l_1,l_2\in[k]$. We consider any weights $\mathbf w_1,\ldots,\mathbf w_k$ that satisfy the requirements in the lemma statement, and additionally satisfy the following: $ w_j(i_1(j))=v_j$ for any $j \in [k]$ with $v_j > 0$, and $w_{l_1}(i_2(l_1))=w_{l_2}(i_3(l_2))=\epsilon$ for some $\epsilon>0$. Denote by $K^\epsilon$ the corresponding cost function, which by assumption is submodular. Next, we consider the following sets:
    \begin{equation*}
        A=\set{i_1(j):j\in[k],\, v_j>0},  \quad  B=A\cup\{i_2(l_1)\}.
    \end{equation*}
    By construction, we have $\mathbf u(A)=\mathbf v$ and $\mathbf u(B) = \mathbf v + \epsilon \mathbf e_{l_1}$. Moreover, note that $i_3(l_2)$ is not contained in either $A$ or $B$ by construction. Then, since $K^\epsilon$ is submodular, we have
    \begin{align*}
        0&\geq (K^\epsilon(B\cup\{i_3(l_2)\})-K^\epsilon(B)) \\
        &\qquad- (K^\epsilon(A\cup\{i_3(l_2)\})-K^\epsilon(A)) \\
        &\overset{(i)}{=} \int_{y=0}^{\epsilon} \int_{z=0}^{\epsilon}  \frac{\partial^2 f}{\partial x_{l_1} \partial x_{l_2}}(\mathbf v+ y\mathbf e_{l_1} + z\mathbf e_{l_2}) dy dz \\
        &\overset{(ii)}{ =} \epsilon^2 \frac{\partial^2 f}{\partial x_{l_1} \partial x_{l_2}}(\mathbf v) + o(\epsilon^2)
    \end{align*}
    In $(i)$ we used \cref{eq:useful_computation_square} and in $(ii)$ we used the fact that $ \frac{\partial^2 f}{\partial x_{l_1} \partial x_{l_2}}$ is continuous since $f$ is twice continuously differentiable. Therefore, this shows that $ \frac{\partial^2 f}{\partial x_{l_1} \partial x_{l_2}}(\mathbf v)\leq 0$. This ends the proof.
\Halmos\end{proof}

\subsection{Preprocessing}\label{appendix:preprocessing}

In this subsection, we prove the following lemma using a constructive argument, which also provides a technique for obtaining the desired solution.

\begin{lemma}\label{lemma: exists_satisfying_properties}
    There exists an optimal solution to \eqref{eq: primal} that satisfies both Property~\ref{property: nested} and Property~\ref{property:LP_solution_reduced}.
\end{lemma}

\begin{proof}
The procedure used to ensure that Properties \ref{property: nested} and \ref{property:LP_solution_reduced} hold involves solving a modified submodular optimization problem. This problem is similar to \eqref{eq: primal}, but also optimizes over an additional variable $z^i_t$, representing the total quantity of item $i\in[N]$ ordered in period $t\in[T]$. In contrast to \eqref{eq: primal}, this formulation prioritizes placing orders earlier in time. Let $C^\star$ denote the optimal objective value of \eqref{eq: primal}. We solve the following problem:
\begin{align}\label{eq: primal_modified}
    \text{minimize}  &  \sum_{t=1}^T  \sum_{i=1}^N t\cdot z_t^i \tag{P2}\\
    \text{subject to}&\sum_{S\subseteq [N] }\sum_{t=1}^T K(S) y^S_t \leq C^\star \nonumber\\
    & z^i_t = \sum_{S: i\in S \subseteq [N]}y_t^S, &&i\in[N],t\in [T] \nonumber\\
    & \sum_{t\in [c,d]} z^i_{t} \geq 1, &&(i,c,d)\in \mathcal D \nonumber\\
    &y_t^S \geq 0, &&t\in [T], S\subseteq [N]\nonumber
\end{align}

The $y$ component of the optimal solution $(y,z)$ to \eqref{eq: primal_modified} is also optimal for \eqref{eq: primal}, since the first constraint guarantees that the order cost is bounded by $C^\star$. Using the procedure from Lemma 1 in \cite{cheung2016submodular}, this $y$ component can be transformed to satisfy the nested structure property. The following claim formalizes this transformation and shows that the resulting solution $\tilde y$ remains optimal for \eqref{eq: primal} while also satisfying both structural properties: \cref{property: nested} and \cref{property:LP_solution_reduced}.

%: for each time period $t\in[T]$, without loss of generality, index the items so that $z^1_t\leq z^2_t \leq \dots\leq z^N_t$, and set $z_t^0:=0$. Define $S(i):=\set{i, i+1, \dots, N}$, and set $\tilde y_t^{S(i)} = z_t^{i}-z_t^{i-1}$ for each $i\in[N]$ and $\tilde y_s^S=0$ for all other $S\subseteq [N]$. 

\begin{myclaim}\label{lemma:check_preprocessing}
    Let $(y,z)$ be an optimal solution to \cref{eq: primal_modified}. 
    For each $t\in[T]$, let $\pi_t(1),\ldots,\pi_t(N)$ be the items ordered so that 
    \[
        z_t^{\pi_t(1)} \leq \ldots \leq z_t^{\pi_t(N)}.
    \]
    For $i\in[N]$, define the subset
    \[
        S(i,t) := \{\pi_t(i),\ldots,\pi_t(N)\}
    \]
    of items whose order quantity is at least $z_t^{\pi_t(i)}$. 
    We then define
    \[
        \tilde y_t^{S(i,t)} := z_t^{\pi_t(i)} - z_t^{\pi_t(i-1)},
    \]
    where $z_t^{\pi_t(0)} := 0$.
    For all other sets $S\subseteq[N]$, set $\tilde y_t^S=0$.
    Then, $\tilde y$ is also an optimal solution to \cref{eq: primal} and satisfies Property~\ref{property: nested} and Property~\ref{property:LP_solution_reduced}. 
\end{myclaim}
\end{proof}

% We now prove that the pre-processed solution satisfies the structural properties we require. 

\begin{proof}[Proof of \cref{lemma:check_preprocessing}.]
    For any $t\in[T]$, $\tilde y_t^S$ has positive value only for sets of the form 
    $S(i,t) = \{\pi_t(i),\ldots,\pi_t(N)\}$ with $i\in[N]$.  
    These sets are nested, so $\tilde y$ satisfies Property~\ref{property: nested}. 
    
    Let $j \in [N]$ denote an actual item label and let $i\in[N]$ be such that $j = \pi_t(i)$. Then
    \[
        \sum_{S: j \in S \subseteq [N]} \tilde y_t^S
        = \sum_{i' \in [i]} \tilde y_t^{S(i',t)}
        = z_t^{\pi_t(i)} 
        = z_t^{j}.
    \]
    Also, for any $t\in[T]$, submodularity of $K$ implies (see \cite[Lemma 1]{cheung2016submodular})
    \[
        \sum_{S\subseteq [N] } K(S) \tilde y_t^S \leq \sum_{S\subseteq [N] } K(S) y^S_t.
    \]
    Thus $(\tilde y,z)$ is feasible for \cref{eq: primal_modified} and achieves the same objective value as $(y,z)$. Since $(y,z)$ is optimal, $(\tilde y,z)$ is also optimal for \cref{eq: primal_modified}.  
    By definition, $C^\star$ is the optimal value of \cref{eq: primal}, so the constraints of \cref{eq: primal_modified} imply $\tilde y$ is also optimal for \cref{eq: primal}. 
    
    It remains to show Property~\ref{property:LP_solution_reduced}.  
    Suppose, for contradiction, that it does not hold. Then there exists a demand point $(j,c,d)$ and two distinct times $t_1,t_2 \in [c,d]$ with $t_1 < t_2$, $z_{t_2}^j > 0$, and
    \begin{equation}\label{eq:contradiction_hypothesis}
        K(A_{t_1}^j \cup\{j\}) - K(A_{t_1}^j) \leq K(B_{t_2}^j) - K(B_{t_2}^j\setminus\{j\}).
    \end{equation}
    By \cref{claim:optimality_condition_dual_variable}, since $\tilde y$ is optimal for \cref{eq: primal}, the inequality in \eqref{eq:contradiction_hypothesis} must hold with equality.  
    Furthermore, since $z_{t_2}^j > 0$, we have $\tilde y_{t_2}^{B_{t_2}^j} > 0$.  
    Now consider decreasing the amount of $j$ ordered at $t_2$ and increasing the amount ordered at $t_1$ by
    \[
        \delta := \min\{\tilde y_{t_2}^{B_{t_2}^j},\ \tilde y_{t_1}^{A_{t_1}^j} + (1 - z_{t_1}^j) \cdot \mathbf{1}[A_{t_1}^j = \emptyset] \}.
    \]
    This strictly improves the objective of \cref{eq: primal_modified} by $(t_2 - t_1)\delta > 0$, contradicting the optimality of $(\tilde y,z)$.  
    Therefore, $\tilde y$ satisfies Property~\ref{property:LP_solution_reduced}.
\Halmos\end{proof}

\subsection{Bounding the expected cost of the randomized procedure}
\label{subsection: appendix_expected_cost}

We now formalize the randomized procedure. 
Let $q_j(S,t;\alpha,\beta)$ denote the inclusion proportion of each subset $(S,t)$ in $\Rect_j(\alpha,\beta;I)$, defined as $1$ if $\alpha_j(S,t)\leq \alpha-1$, $\beta$ if $\alpha_j(S,t)=\alpha$, and $0$ otherwise. Then, the area of $\Rect_j(\alpha,\beta;I)$ is
    \begin{align*}
        \Area_j & \paren{\Rect_j(\alpha,\beta;I)}\\
        &:= \int_{t\in I} \paren{\sum_{S\subseteq[N]} y_t^S q_j(S,t;\alpha,\beta)} dt.
    \end{align*}
    
Fix a category $j \in [k]$. For each level-set $\Rect_j(\alpha,\beta;I)$, we define a probability distribution $\Qcal_j(\alpha,\beta;I)$ supported on pairs $(S_r(t), t)$ with $t \in I$ and $r \in [n(t)]$. Intuitively, this distribution samples subsets proportionally to their (possibly fractional) contribution within the level-set.
Formally, for any $t \in I$ and $S \subseteq [N]$, the density of $\Qcal_j(\alpha,\beta;I)$ at $(S,t)$ is defined as:
\begin{align*}
    &p_{\Qcal_j(\alpha,\beta;I)}(S,t) \\
    &:= \frac{y_t^S \paren{ q(S,t;\alpha,\beta)-q(S,t;\alpha_p,\beta_p) }}{\Area_j(\Rect_j(\alpha,\beta;I)) - \Area_j(\Rect_j(\alpha_p,\beta_p;I))}.
\end{align*}
That is, the probability of sampling $(S,t)$ is proportional to its associated fractional order quantity $y_t^S$, scaled by $q(S,t;\alpha,\beta)-q(S,t;\alpha_p,\beta_p)$ to account for the portion of this fractional subset contained in the level-set $\Rect_j(\alpha,\beta;I)$ but not in its parent. The denominator is given by the area of $\Rect_j(\alpha,\beta;I)$ outside its parent.

To satisfy $\Rect_j(\alpha,\beta;I)$, we sample $(S,t)\sim\Qcal_j(\alpha,\beta;I)$, add the corresponding order to the schedule, and delete all level-sets in higher layers that contain time $t$. By construction, this procedure ensures that every level-set in every layer for category $j$ is satisfied by at least one order. The randomized algorithm is summarized in Algorithm~\ref{alg:randomized}. 

\begin{algorithm}[t]

\caption{Randomized version of Step 2 of Algorithm~\ref{alg:water_filling}}\label{alg:randomized}

%\SetAlgoLined
\LinesNumbered
\everypar={\nl}

% \hrule height\algoheightrule\kern3pt\relax
\KwIn{Category $j\in[k]$. Constructed layers $\Lcal_l(j)$ for $l\in[L]$ where $L$ is the total number of layers}
\KwOut{Schedule of orders $\tilde\Ocal^{(j)}$}

{\nonl \textbf{Step 2: Order placements for each level-set}}

Initialize all regions remaining to be satisfied $\Rcal_l\gets \Lcal_l(j)$ for $l\in[L]$

Initialize orders $\tilde\Ocal^{(j)}\gets \emptyset$

\For{$l=1,\ldots,L$}{
    \For{$\Rect_j(\alpha,\beta;I)\in \Rcal_l$}{
        Sample $(S,t)\sim \Qcal_j(\alpha,\beta;I)$ and update orders $\tilde\Ocal^{(j)}\gets \tilde\Ocal^{(j)}\cup\{(S,t)\}$

        For every layer $l'>l$, delete from $\Rcal_{l'}$ all level-sets $\Rect_j(\alpha',\beta';I')$ with $t\in I'$
    }
}

\Return $\tilde\Ocal^{(j)}$

% \hrule height\algoheightrule\kern3pt\relax
\end{algorithm}

We now formalize notation for the layers. Each layer $\Lcal_l$ consists of $N_l:=|\Lcal_l|$ disjoint level-sets, denoted $(R_r^l)_{r\in[N_l]}$, ordered by increasing time intervals. Each level-set \( R_r^l \in \Lcal_l \) is associated with a parent level-set in the previous layer \( \Lcal_{l-1} \), from which the time interval was split into two; one of these subintervals corresponds to \( R_r^l \). We now define notation to specify the parent and ancestors of a given level-set. Recall that the parent of any level-set in the first layer \( \Lcal_1 \) is \( \Rect_j(0,1; (0, T]) \).

\begin{definition}[Parents and ancestors]
Given a level-set $R^l_r\in\Lcal_l$, we denote by $\pi(r,l)\in[N_{l-1}]$ the index of its parent level-set $R_{\pi(r,l)}^{l-1}$. More generally, we define $\pi(r, l, l')\in[N_{l'}]$ to be the index of the ancestor of $R_r^l$ at layer $l'\in[l-1]$. For example, $\pi(r, l) = \pi(r, l, l-1)$ and $\pi(r,l,l-2)=\pi(\pi(r,l),l-1)$, and so on.
\end{definition}

We are now ready to bound the expected cost of the sampled schedule. In particular, we prove that the expected ordering cost for each of the $k$ categories is at most $4C^\star$, where $C^\star$ is the optimal objective value of the LP relaxation in \eqref{eq: primal}.

\begin{proof}[Proof of \cref{prop:cost_each_order}]
    Since we focus on a single category $j\in[k]$ throughout this proof, we omit subscripts and superscripts referring to category $j$. 
    In Algorithm~\ref{alg:randomized}, each level-set is satisfied exactly once in the sense that either it is first satisfied by an ancestral sample, or the algorithm samples once from its residual region.
    % exactly one order is placed to satisfy each level-set $R_r^l$ for all layers $l\geq 1$ and level-sets $r\in[N_l]$.

For each such level-set, we define a corresponding subregion $Q_r^l$, which intuitively captures the portion of $R_r^l$ not covered by its parent level-set. Formally, let $R_r^l=\Rect_j(\alpha,\beta;I)$ and let $R_{\pi(r,l)}^{l-1}=\Rect_j(\alpha_p,\beta_p;I_p)$ be its parent level-set, where $I\subseteq I_p$. We define $Q_r^l $ to consist of a $q(S',t;\alpha,\beta) - q(S',t;\alpha_p,\beta_p)$ fraction of each subset $S':=S_{r'}(t)$, for all $t\in I$ and $r'\in[n(t)]$. Accordingly, the area of subregion $Q_r^l$ is given by $\Area_j(Q_r^l)=\Area_j(R_r^l)-\Area_j(R_r^l \cap R_{\pi(r,l)}^{l-1})$. Importantly, unlike the level-sets $R_r^l$, the sub-regions $(Q_r^l)_{l\geq 1,r\in[N_l]}$ are disjoint. 
    
    For any region $R$, let $p_{alg}(R)$ and $p_{LP}(R):= \Area_j(R)$ denote the probability that Algorithm~\ref{alg:randomized} and the pre-processed LP solution $(y_t^S)_{t\in[T],S\subseteq [N]}$ to the DICP relaxation \eqref{eq: primal} place an order in $R$, respectively. Similarly, let $\text{cost}_{alg}(R)$ and $\text{cost}_{LP}(R)$ denote the total expected cost of orders from $R$ under Algorithm~\ref{alg:randomized} and the LP solution, respectively. Whenever Algorithm~\ref{alg:randomized} orders from a region $Q_r^l$, it samples this order with probabilities proportional to the LP solution. Hence,
    \begin{equation*}
        \text{cost}_{alg}(Q_r^l) = \text{cost}_{LP} (Q_r^l)\cdot \frac{p_{alg}(Q_r^l)}{p_{LP}(Q_r^l)}
    \end{equation*}
when $p_{LP}(Q_r^l) > 0$.

    The benefit of this sampling procedure is that the probability ratio is bounded:

    \begin{lemma}\label{lemma:bounded_ratio_probas}
        For any layer $l\geq 1$ and region $r\in[N_l]$, we have $p_{alg}(Q_r^l) = 4 \cdot p_{LP}(Q_r^l)$.
    \end{lemma}

Using this, we can bound the total expected cost of the algorithm:

    \begin{align*}
        \Ebb\sqb{\sum_{(S,t)\in\tilde\Ocal^{(j)}} K(S)} &= \sum_{l\geq 1}\sum_{r\in[N_l]} \text{cost}_{alg}(Q_r^l) \\
        &= 4 \sum_{l\geq 1}\sum_{r\in[N_l]} \text{cost}_{LP}(Q_r^l) \\
        &\leq 4C^\star.%\cdot\text{cost(LP)}.
    \end{align*}
    In the last inequality, we used the fact that all regions $Q_r^l$ for $l\geq 1$ and $r\in[N_l]$ are disjoint. 
\Halmos\end{proof}

We now prove the lemma bounding the probability that the algorithm places an order from each region $Q_r^l$ for $l\geq 1$ and $r\in[N_l]$.

\begin{proof}[Proof of \cref{lemma:bounded_ratio_probas}]
    We prove the result by induction on the layer index $l\geq 1$. In the first layer $l=1$, Algorithm~\ref{alg:randomized} always samples an order from each of the two regions $Q_r^l=R_r^l$ for $r\in[2]$. Hence,
    \begin{equation*}
        p_{alg}(Q_r^1) = 1 = 4\cdot  p_{LP}(Q_r^1),\quad r\in[2].
    \end{equation*}
    In the last equality, we used the fact that each constructed region has area $1/4$ by design.

    Next, fix $l\geq 2$ and $r\in[N_l]$, and assume that the claim holds for all level-sets in all layers $l'<l$. By construction, Algorithm~\ref{alg:randomized} samples from $Q_r^l$ if and only if region $R_r^l$ has not been satisfied by any order sampled in previous layers. That is, sampling from $Q_r^l$ occurs if and only if there were no order from the region $R_r^l\cap R_{\pi(r,l)}^{l-1}$ while processing earlier layers. In particular, we have
    \begin{align}\label{eq:disjoint_regions_partition}
    \begin{split}
        R_r^l &= Q_r^l \sqcup (R_r^l\cap R_{\pi(r, l)}^{l-1}) \\
        &= Q_r^l \sqcup \bigsqcup_{l'\in[l-1]} (R_r^l \cap Q_{\pi(r, l, l')}^{l'}), 
    \end{split}
    \end{align}
    where $\bigsqcup$ denotes a disjoint union. Therefore,
    \begin{align*}
        &p_{alg}(Q_r^l) \\
        &= 1- \Pbb\sqb{\bigcup_{ l'\in[l-1]} \{\text{Alg}~\ref{alg:randomized}\text{ sampled in }R_r^l \cap Q_{\pi(r, l, l')}^{l'}\}}\\
        &\overset{(i)}{=} 1- \sum_{l'\in[l-1]} \Pbb\sqb{ \text{Alg}~\ref{alg:randomized}\text{ sampled in }R_r^l \cap Q_{\pi(r, l, l')}^{l'} }\\
&=
1-
\sum_{l'\in[l-1]}
p_{alg}(Q_{\pi(r,l,l')}^{l'})
\notag\\
&\qquad \times
\Pbb\Big[
\text{Alg}~\ref{alg:randomized}
\text{ sampled in }
R_r^l \cap Q_{\pi(r,l,l')}^{l'}
\;\Big|\;
\notag\\
&\qquad\qquad
\text{Alg}~\ref{alg:randomized}
\text{ sampled in }
Q_{\pi(r,l,l')}^{l'}
\Big]\\
        &\overset{(ii)}{=} 1- \sum_{l'\in[l-1]} p_{alg}(Q_{\pi(r, l, l')}^{l'})\cdot \frac{p_{LP}(R_r^l \cap Q_{\pi(r, l, l')}^{l'})}{p_{LP}(Q_{\pi(r, l, l')}^{l'})} \\
        &\overset{(iii)}{=} 1- 4 \sum_{l'\in[l-1]} p_{LP}(R_r^l \cap Q_{\pi(r, l, l')}^{l'}) \\
        &\overset{(iv)}{=} 1-4(p_{LP}(R_r^l) - p_{LP}(Q_r^l)) \\
        &\overset{(v)}{=} 4\cdot p_{LP}(Q_r^l).
    \end{align*}
    In $(i)$, we used the fact that the events in the union are disjoint: Algorithm~\ref{alg:randomized} samples an order from region $Q_{\pi(r, l, l')}^{l'}$ only if $R_{\pi(r, l, l')}^{l'}$ has not been satisfied, which implies that for all $m\in[l'-1]$, no order was sampled in $R_r^l \cap Q_{\pi(r, l, m)}^{m}$. In $(ii)$ we used the fact that Algorithm~\ref{alg:randomized} samples in $R_{\pi(r, l, l')}^{l'}$ proportionally to the LP solution. In $(iii)$ we used the induction hypothesis, and in $(iv)$ we used \cref{eq:disjoint_regions_partition}. Finally, in $(v)$ we used the fact that $p_{LP}(R_r^l)=\Area_j(R_r^l)= 1/4$. This completes the induction and the proof.
\Halmos\end{proof}

\section{Proof of correctness for the SJRP approximation algorithm}\label{section: appendix Reduction}

This appendix provides the proofs omitted from the main text that establish the correctness of the SJRP algorithm, assuming access to the approximation algorithm for the DICP. We use the same notations as in the construction of $\Ical_{DICP}$ described in \cref{subsection: Reduction alg}. In particular, we let $(x,y)$ be the optimal solution to (LP)---the linear relaxation of $\Ical$ that was used to construct $\Ical_{DICP}$. 

% Let $\tilde O_t$ for $t\in[T]$ denote a feasible integral schedule of orders for $\Ical_{DICP}$. We denote by $O_t$ for $t\in[T]$ the corresponding schedule of orders for $\Ical$ as defined in \cref{eq:construction_final_schedules}.

Let $\tilde O_m$ for $m\in[M]$ denote a feasible integral schedule of orders for $\Ical_{DICP}$. We let $\widehat O_t=\bigcup_{m:t(m)=t}\tilde O_m$ for $t\in[T]$, and denote by $O_t$ the corresponding schedule of orders for $\Ical$ as defined in \cref{eq:construction_final_schedules}.

\begin{proof}[Proof of \cref{lemma: Reduction_feasible}]
Fix a demand point $(i,d)$. Define $z$ as in the construction of $\Ical_{DICP}$: that is, $z_s^i = \sum_{S:i\in S} y_s^S$. Since the solution $y$ is feasible for (LP), we have
    \begin{align*}\label{eq:useful_eq_int}
       \int_{s=0}^d  z_s^i ds &\geq  \int_{s=0}^{s^\star_{(i,d)}}  z_s^i ds + \int_{s=s^\star_{(i,d)}}^d  x_s^i ds \\
       &\geq \int_{s=0}^{s^\star_{(i,d)}}  z_s^i ds + 1.
    \end{align*}
    In particular, there must exist an integer $\kappa\geq 1$ such that
    \begin{equation*}
        \int_{s=0}^{s^\star_{(i,d)}}  z_s^i ds \leq \frac{\kappa-1}{2} \leq \frac{\kappa}{2} \leq \int_{s=0}^{d}  z_s^i ds.
    \end{equation*}
    Next, recall that $a_{\kappa-1}^i$ is defined as the largest value $\theta$ for which $\int_{s=0}^{\theta}  z_s^i ds=\frac{\kappa-1}{2}$ and similarly, $b_{\kappa}^i$ is the smallest value $\theta$ for which $\int_{s=0}^{\theta}  z_s^i ds=\frac{\kappa}{2}$. In summary, we showed that
    \begin{equation*}
        s^\star_{(i,d)} \leq a_{\kappa-1}^i \leq b_\kappa^i \leq d.
    \end{equation*}
    This implies \cref{eq:important_property_Reduction}: $(a_{\kappa-1}^i,b_\kappa^i]\subseteq (s_{(i,d)}^\star,d]$.

% By the construction of the continuous-time extension, if $x_{sd}^i>0$ for a given $s\in[T]$, then $x_{rd}^i>0$ for all $r\in(s-1, s]$. Hence, we have $s_{(i,d)}^\star = s_{(i,d)}-1$, where $s_{(i,d)}:=\min\{s\in[T]: x_{sd}^i>0\}$. As a result, in discrete time we have $\sqb{\lfloor a_{\kappa-1}^i\rfloor+1,\lceil b_\kappa^i \rceil} \subseteq \sqb{s_{(i,d)},d}$. Since $x_{sd}^i$ is nonzero at time $s=s_{(i,d)}$, this shows that $[s_{(i,d)},d]$ is included in the active interval of $(i,d)$. Hence, $\sqb{\lfloor a_{\kappa-1}^i\rfloor+1,\lceil b_\kappa^i \rceil}$, which is a demand interval in $\Ical_{DICP}$, is also included in the active interval of $(i,d)$. This shows that the final schedule $(O_s)_{s\in[T]}$ satisfies the demand point $(i,d)$ and ends the proof.
By construction of the reduced DICP instance, the interval $(a_{\kappa-1}^i,b_\kappa^i]$ is exactly the union of the atoms $J_m$ for $m\in[c_\kappa^i,d_\kappa^i]$. Since the DICP solution is feasible, item $i$ is ordered in some atom $m\in[c_\kappa^i,d_\kappa^i]$. Therefore $J_m\subseteq(a_{\kappa-1}^i,b_\kappa^i]\subseteq(s^\star_{(i,d)},d]$. Since $J_m\subseteq(t(m)-1,t(m)]$, this implies $t(m)\leq d$. Hence $i\in\widehat O_{t(m)}$ for some $t(m)\leq d$, and the final schedule $(O_s)_{s\in[T]}$ satisfies demand point $(i,d)$.
\Halmos\end{proof}

Below, we prove \cref{lemma: Reduction_holding}, which shows that the holding cost incurred by the schedule ${O_t}$ for ${t \in [T]}$ is at most the cost of the solution $(x, y)$ to (LP). This result also appears in \cite{cheung2016submodular}, but we include its proof for completeness.

\begin{proof}[Proof of \cref{lemma: Reduction_holding}]
    We write the dual program of (LP). Let $M^i_d$ and $L^i_{sd}$ be the dual variables corresponding to the first and second constraints in the linear relaxation of \cref{eq: IP} respectively. Then the dual program of (LP) is 
\begin{align*}
    \text{maximize}\quad&\sum_{(i,d)\in \mathcal D} M_d^i\tag{D}\\
    \text{subject to}\quad& M_d^i\leq q_{id}h^i_{sd} + L_{sd}^i, &&(i,d)\in \mathcal D, s\in[d]\nonumber\\
    &\sum_{i\in S}\sum_{t=s}^T L_{st}^i\leq K(S), &&s\in[T], S\subseteq [N]\nonumber\\
    &L_{sd}^i\geq 0, &&(i,d)\in \mathcal D, s\in[d]\nonumber
\end{align*}

By complementary slackness in the optimal (LP) solution, any variable $x^i_{sd} > 0$ implies that the corresponding constraint $M_d^i \leq q_{id}h^i_{sd} + L_{sd}^i$ is tight. Since $L_{sd}^i \geq 0$, it follows that $M_d^i \geq q_{id}h^i_{sd}$ for all $i, s, d$ such that $x^i_{sd} > 0$. In particular, due to the monotonicity of $h^i_{sd}$ in $s$, this holds for every $(i,d)\in\mathcal{D}$ and every $s \geq \min\{s\in[d]:x_{sd}^i>0\}=:s_{(i,d)}$.

We now bound the holding cost incurred by the schedule of orders $O_t$ for $t\in[T]$. By \cref{lemma: Reduction_feasible} and its proof, each demand point $(i,d)$ is fulfilled by an order placed at some original period $s\geq s_{(i,d)}$, where $s_{(i,d)}:=\min\{s\in[d]:x_{sd}^i>0\}$. Therefore, the holding cost incurred for each $(i,d)$ is $q_{id}h^i_{sd}\leq M_d^i$. Summing over all demand points, the total holding cost is at most $\sum_{(i,d) \in \mathcal{D}} M_d^i$, which is itself bounded by the cost of (LP) by weak duality.

\Halmos\end{proof}

Finally, we prove the bound on the ordering cost:

% We consider the solution $\tilde y$, defined as follows:
%     \begin{align*}
%         \tilde y_s^S := 2 y_s^S, \quad  s\in[T],S\subseteq [N].
%     \end{align*}
%     We now show that this is feasible for $\Ical_{DCIP}$. For convenience, consider the continuous extension $\tilde y_{r}^S:=\tilde y_s^S$ for any $r\in(s-1,s]$ and $s\in[T]$. Then, for any $i\in[N]$ and $\kappa\in[K^i]$, we check that the constraint for $\Ical_{DCIP}$ is satisfied: 
%     \begin{align*}
% \sum_{s=\floor{a_{\kappa-1}^i}+1}^{\ceil{b_\kappa^i}} \,\sum_{S:i\in S} \tilde y_s^S &= \sum_{s=\floor{a_{\kappa-1}^i}+1}^{\ceil{b_\kappa^i}} \,
% \int_{r=s-1}^s \paren{\sum_{S:i\in S} \tilde y_r^S} dr\\
% &\geq\int_{r=a_{\kappa-1}^i}^{b_\kappa^i} \paren{ \sum_{S:i\in S} \tilde y_r^S }dr \\
% &= \int_{r=a_{\kappa-1}^i}^{b_\kappa^i} 2z_{r}^i dr \overset{(i)}{=}1.
%     \end{align*}
%     In $(i)$ we used the definition of $a_{\kappa-1}^i$ and $b_\kappa^i$ as $\frac{\kappa-1}{2}$ and $\frac{\kappa}{2}$ quantiles of $z_{s}^i$ respectively.
%     Hence, $\tilde y$ is feasible for $\Ical_{DICP}$, and it has ordering cost at most 2 times the ordering cost of $(x,y)$ by construction.
% \Halmos\end{proof}

\begin{proof}[Proof of \cref{lemma: Reduction_ordering}]
We define a fractional solution $\tilde y$ for the LP relaxation of $\Ical_{DICP}$ by setting
\[
\tilde y_m^S:=2|J_m|y_{t(m)}^S,\qquad m\in[M],\ S\subseteq[N].
\]
We first show that $\tilde y$ is feasible for $\Ical_{DICP}$. Fix $i\in[N]$ and $\kappa\in[K^i]$. Since $(a_{\kappa-1}^i,b_\kappa^i]=\bigcup_{m=c_\kappa^i}^{d_\kappa^i}J_m$, we have
\begin{align*}
    \sum_{m=c_\kappa^i}^{d_\kappa^i}\sum_{S:i\in S}\tilde y_m^S
&=
\sum_{m=c_\kappa^i}^{d_\kappa^i}2|J_m|\sum_{S:i\in S}y_{t(m)}^S\\
&
=
\int_{a_{\kappa-1}^i}^{b_\kappa^i}2z_r^i\,dr\\
&
=
1,
\end{align*}
where the last equality follows from the definition of $a_{\kappa-1}^i$ and $b_\kappa^i$ as the $(\kappa-1)/2$ and $\kappa/2$ quantiles of $z_s^i$. Hence, $\tilde y$ is feasible for the LP relaxation of $\Ical_{DICP}$.

The ordering cost of $\tilde y$ is
\begin{align*}
    \sum_{m=1}^M\sum_{S\subseteq[N]}K(S)\tilde y_m^S
&=
2\sum_{m=1}^M |J_m|\sum_{S\subseteq[N]}K(S)y_{t(m)}^S\\
&=
2\sum_{t=1}^T\sum_{S\subseteq[N]}K(S)y_t^S,
\end{align*}
where the last equality uses the fact that the atoms contained in each original period $(t-1,t]$ partition that period. Therefore, the ordering cost in the LP relaxation of $\Ical_{DICP}$ is at most twice the ordering cost of the optimal LP solution $(x,y)$ to $\Ical$.
\Halmos\end{proof}

\begin{proof}[Proof of \cref{lemma: atom_mapping_cost}]
For each original period $t\in[T]$, Algorithm~\ref{alg:Reductions} first defines
\[
\widehat O_t:=\bigcup_{m:t(m)=t}\tilde O_m.
\]
By subadditivity of the monotone submodular function $K$,
\[
K(\widehat O_t)\leq \sum_{m:t(m)=t}K(\tilde O_m).
\]
Since the final schedule satisfies $O_t\subseteq \widehat O_t$ for every $t\in[T]$, monotonicity implies
\[
\sum_{t=1}^T K(O_t)
\leq
\sum_{t=1}^T K(\widehat O_t)
\leq
\sum_{m=1}^M K(\tilde O_m).
\]
\end{proof}

%%%%%%%%%%%%%%%%%
\end{document}